\documentclass[12pt,a4paper,oneside]{article}
\usepackage[english]{babel}     
\usepackage[latin1]{inputenc}   
\usepackage[T1]{fontenc}
\usepackage{geometry}
\usepackage{amsmath}
\usepackage{bm}
\usepackage{bbm}
\usepackage{amsfonts}
\usepackage{url}
\usepackage{rotating}           
\usepackage{longtable}          
\usepackage{multirow}           
\usepackage{multicol}           
\usepackage{booktabs} 
\usepackage{cancel}
\usepackage{nicefrac}           
\usepackage[authoryear]{natbib}
\usepackage{setspace}
\usepackage[title]{appendix}

\usepackage{verbatim}           
\usepackage{graphicx}
\usepackage{lscape}

\usepackage{wasysym}            
\usepackage{marvosym}           
\usepackage{amssymb}            
\usepackage{amsthm}
\usepackage{pifont}
\usepackage{fancyvrb}
\usepackage{tikz}
\usepackage{csquotes}	
\usepackage{booktabs}
\usepackage[flushleft]{threeparttable}	

\usepackage{pbox}
\usepackage{longtable}

\usepackage{chngcntr}
\counterwithin{figure}{section}
\counterwithin{table}{section}
\counterwithin{equation}{section}

\usepackage{float}
\restylefloat{table}
\restylefloat{figure}
\usepackage{float}
\restylefloat{table}

\usepackage{chngcntr}
\counterwithout{figure}{section}
\counterwithout{table}{section}

\usepackage[colorlinks, allcolors=blue]{hyperref}

\usepackage{subfigure}
\usepackage{diagbox}
\usepackage{algorithm} 
\usepackage{algpseudocode} 
\makeatletter
\newcommand{\algmargin}{\the\ALG@thistlm}
\makeatother
\newlength{\whilewidth}
\algnewcommand{\parState}[1]{\State%
	\parbox[t]{\dimexpr\linewidth-\algmargin}{\strut #1\strut}}

\author{Helmut Farbmacher, Rebecca Groh, Michael M\"uhlegger, Gabriel Vollert\thanks{farbmacher@tum.de} \\ \small Technical University of Munich, Germany \vspace{0.3cm}}

\title{Revisiting the Many Instruments Problem using Random Matrix Theory}

\newcommand{\lam}{(-\lambda)}

\newcommand{\ZZlam}{(Z'Z/n+\lambda I_k)}
\newcommand{\ZZlamout}{(ZZ'/n+\lambda I_n)}

\newcommand{\E}{\textrm{E}}
\newcommand{\V}{\textrm{Var}}

\newtheorem{theorem}{Theorem}[section]

\newtheorem{corollary}{Corollary}[theorem]
\newtheorem{lemma}[theorem]{Lemma}

\theoremstyle{definition}

\newtheorem{assumption}{Assumption}

\renewenvironment{proof}{{\textit{Proof}.}}{\hfill\qedsymbol \bigskip}
\usepackage{amsmath}

\DeclareMathOperator*{\argmin}{arg\,min}
\DeclareMathOperator{\tr}{tr}
\usepackage{amssymb}

\begin{document}

\newcommand*{\thisdraft}{This version: \today}
\newcommand*{\firstdraft}{\small{(First version: August 15, 2024)}}

\date{\thisdraft \\ \firstdraft}

\maketitle
	
\begin{abstract}	
	\noindent	
	Instrumental variables estimation with many instruments is biased. Traditional bias-adjustments are closely connected to the Silverstein equation. Based on the theory of random matrices, we show that Ridge estimation of the first-stage parameters reduces the implicit price of bias-adjustments. This leads to a trade-off, allowing for less costly estimation of the causal effect which comes along with improved asymptotic properties. Our theoretical results nest existing ones on bias approximation and adjustment with ordinary least-squares in the first-stage regression and, moreover, generalize them to settings with more instruments than observations. Finally, we derive the optimal tuning parameter of Ridge regressions in simultaneous equations models, which comprises the well-known result for single equation models as a special case with uncorrelated error terms.
\end{abstract}
	
	
\newpage
	
\section{Introduction}

In this paper, we propose a new bias-adjusted two-stage least-squares (2SLS) Ridge estimator for a simultaneous equations model with a single endogenous regressor and many instrumental variables (IV). Our motivation comes from recent results in the theory of random matrices (RMT), which give a new understanding of the bias induced by many instruments. 

It is well-known that IV estimators perform poorly under many and/or weak instruments. In these settings, IV estimators are biased towards the ordinary least-squares (OLS) estimate \citep[e.g.,][]{bound1995problems,staiger1997instrumental,stock2000gmm,chao2005consistent}. The standard IV estimator is a 2SLS procedure, where we first use least-squares to predict the endogenous variable and use these predictions in a second least-squares regression to estimate the causal effect. Thus, the accuracy of 2SLS estimation depends crucially on the prediction quality in the first-stage regression. We would expect a method that reduces the predictive risk in the first stage to perform better in the second-stage regression. This reasoning motivated a series of papers to replace OLS in the first stage with statistical learning tools (such as \cite{belloni2012sparse}, \cite{gold2020inference} and \cite{angrist2022machine} for Lasso regressions or \cite{okui2011instrumental} and \cite{hansen2014instrumental} for Ridge regressions). In this paper, we are interested in settings with dense first-stage parameters such that Ridge is the preferred choice in the first-stage regression. 

Bias-adjusted 2SLS estimation can be motivated by the Silverstein equation. This link to random matrix theory allows us to improve IV estimation even in settings where no closed-form expression of the Silverstein equation exists. We first consider a standard many IV setting where the number of instruments is large relative to, but still smaller than, the sample size. In this setting, our bias approximation nests the well-known results for least-squares in the first stage \citep[e.g.,][]{hahn2002new,hahn2002notes,bun2011comparison}, and a special case of our bias-adjustment coincides with the well-known bias-adjusted 2SLS estimator \citep[][]{nagar1959bias,chao2005consistent,donald2001choosing}. Bias-adjustments are costly in the sense that they reduce the overall signal strength of the instruments. We show that using Ridge in the first stage reduces the implicit price of bias-adjustments. This leads to a trade-off, which allows for less costly estimation of the causal effect parameter and comes along with improved asymptotic properties in the second-stage regression. Additionally, the results mentioned above generalize to settings where the number of instruments is larger than the number of observations, which has yet to be discussed in the IV literature. Finally, we derive the optimal tuning parameter of Ridge regressions in simultaneous equations models, which comprises the well-known results for single equation models as special case with uncorrelated error terms.

To derive our results, we make the following assumptions. We assume that each instrument has only a small, independent random effect on the endogenous variable. This can be considered an average-case analysis over dense parameters. \cite{chamberlain2004random} also use a random effects assumption to derive a quasi-maximum likelihood estimator in an instrumental variables regression with many instruments. Moreover, we assume a high-dimensional asymptotic regime for the instrumental variables. In such a setting, we would expect Ridge to outperform OLS with respect to prediction risk \citep[e.g.,][among others]{dobriban_wager}. While being standard in the RMT literature, these assumptions are restrictive. However, they allow us to derive insightful results that nest existing approaches and may serve as a basis for future conceptual developments.

Closest to our study in the RMT literature are \cite{dobriban_wager} and \cite{hastie2022surprises} who analyze the performance of Ridge in a single equation model. In the econometrics literature, our study is close to \cite{hahn2002new, hahn2002notes} and \cite{bun2011comparison}, who analyze the bias of 2SLS under many and/or weak instruments. Further, our study contributes to several econometric approaches that adjust for the 2SLS bias \citep[e.g.,][]{nagar1959bias,chamberlain2004random,chao2005consistent,hansen2014instrumental,spiess2017bias}. In economics, our study is closely related to \cite{angrist2022machine}, who compare bias-adjusted IV estimators with estimators based on Lasso in the first stage regression. They conclude that statistical learning techniques are less encouraging in IV settings. In contrast to this conclusion, we show that bias-adjusted 2SLS estimators can indeed be motivated by results from the statistical learning literature. These allow us to improve the bias-adjustment, which leads to a better finite-sample performance. Our computational results can be replicated using software available from \href{https://github.com/farbmacher/RMT-many-instruments}{https://github.com/farbmacher/RMT-many-instruments}


\section{Model and 2SLS Estimation}

We have a random i.i.d. sample $\left\{ y_i,x_i,z_i' \right\} _{i=1}^{n}$ drawn from the following simultaneous equations model
\begin{eqnarray}\label{model}
	y_i&=&\beta x_i+ \varepsilon_i \\
	x_i&=&z_i' \pi + \nu_i
\end{eqnarray}
where $y_i$ and $x_i$ are the entries of the $n\times 1$ vectors $y$ and $x$, and $z_i$ are arranged as the rows of the $n \times k$ matrix $Z$. The aim is to estimate the causal effect $\beta$. Throughout, we assume that the following model conditions hold.
\begin{assumption}[Model]
	\label{A_model}
	Let $\eta_i=\left( \varepsilon_i~\nu_{i}\right)
	^{\prime }$ be normally distributed independent of the other random quantities with $E\left[ \eta_i \right] =0$ and $E\left[ \eta_i \eta_i' \right] =\left[ 
	\begin{array}{cc}
		\sigma_\varepsilon^2 & \sigma_{\varepsilon\nu} \\ 
		\sigma_{\varepsilon\nu} & \sigma_\nu^2%
	\end{array}%
	\right]$, a positive definite matrix.
\end{assumption}

Assumption \ref{A_model} imposes homoskedasticity and joint normality of $\epsilon$ and $\nu$. If $\sigma_{\varepsilon\nu}\neq 0 $,  $x$ is endogenous,  and the OLS estimator of $\beta$ is inconsistent.  The instrumental variables $Z$ can be used to overcome the endogeneity. The standard IV estimator is a 2SLS procedure. First, we use least-squares to predict $\widehat{x}=P_0 x$ with $P_0=Z(Z'Z)^{-1}Z' \, $. In a second least-squares regression, we use $\widehat{x} \, $ to estimate the causal effect $\beta$. A first-stage prediction based on Ridge is defined as $\widehat{x}_\lambda=P_\lambda x$ with $P_{\lambda}=Z(Z'Z/n+\lambda I_k)^{-1}Z'/n \, $ being the Ridge smoother and $\lambda$ being a tuning parameter. When we set $\lambda\downarrow 0$ in settings with fewer instruments than observations, we get 2SLS. If we set $\lambda\downarrow 0$ in settings with more instruments than observations, we obtain a ridgeless 2SLS estimator.\footnote{For a detailed discussion of ridgeless least-squares and its interpretation as min-norm least-squares estimator see, for example, \cite{hastie2022surprises}.}

Following the RMT literature, we consider a high-dimensional asymptotic regime for the instrumental variables. 
\clearpage 
\begin{assumption}[High-Dimensional Asymptotics]\label{A_hda}
	\phantom{The following conditions hold:}
	\begin{itemize}
		\item[a)] The data $Z\in \mathbb{R}^{n\times k}$ are generated as $Z=W\Sigma^{1/2}$ for an $n\times k$ matrix $W$ with i.i.d. entries satisfying $\E[W_{ij}]=0$ and $\V[W_{ij}]=1$, and a deterministic $k\times k$ positive semidefinite covariance matrix $\Sigma \, $.
		\item[b)] The sample size $n \to \infty$ and the number of instruments $k\to \infty$, such that the ratio $k/n \to \gamma >0 \,$.
		\item[c)] The distribution of the eigenvalues of $\Sigma$, $F_{\Sigma}$, also called spectral distribution, converges to a limit probability distribution $H$ supported on $[0, \infty)$, called the population spectral distribution (PSD).
	\end{itemize}
\end{assumption}
See \cite{dobriban_wager} for a detailed discussion of covariance matrices $\Sigma$ that are aligned with this setting. For example, the AR-1 model with $\Sigma_{ij}=\rho^{\vert i-j \vert}$, which is often used in econometrics, satisfies Assumption \ref{A_hda}. 

Moreover, following the RMT literature, we assume that each instrument has only a small and independent random effect on the endogenous variable. 
\begin{assumption}[Random Coefficients]
	\label{A_randcoeff}
	The first-stage parameters $\pi \in \mathbb{R}^k$ are random with $\E[\pi]=0 \,  $, $\V[\pi]=k^{-1}\alpha^2 I_k$ and $\alpha^2\neq 0 \, $.
\end{assumption}
This assumption has been used before in the literature on many instruments, e.g. in a theoretical analysis without RMT by \cite{chamberlain2004random}, and can be regarded as an average-case analysis over dense parameters. The concentration parameter, $\mu^2\equiv n\alpha^2/\sigma_\nu^2 \, $ is a standard measure of instrument strength. In a weak instrument setting, the concentration parameter is small relative to the sample size. 

Throughout the paper, we illustrate our theoretical results using Monte-Carlo simulations based on a standard many IV setting
\begin{eqnarray*}
	y&=&\beta x+ \varepsilon \\
	x&=&Z\pi + \nu
\end{eqnarray*}
with $\varepsilon=\rho \nu +\sqrt{1-\rho^2} w$ where $\nu\sim N(0,1)$ and $w\sim N(0,1)$ are independent, and where we generally set $\rho=0.6$. Moreover, $\pi\sim N(0,\sigma^2_\pi)$ with $\sigma^2_\pi=\alpha^2/k$ and $\alpha^2=\gamma F \, $. This allows us to control the first-stage $F$ statistic, which we generally set equal to 5. We show simulation results for $Z\sim N(0,\Sigma)$ and a) isotropic $\Sigma=I_k$ or b) AR-1 model $\Sigma_{ij}=0.5^{\vert i-j\vert}$. Additionally, we set $n=200$, $k=150$ or $k=250$, and $\beta=0$.

\section{Results from Random Matrix Theory}

We start with a brief review of key results from the theory of random matrices, which help us to describe and adjust for the bias induced by Ridge estimation in the first-stage regression. From \cite{marchenko1967distribution} and \cite{silverstein1995strong} we know that the spectral distribution of the sample covariance matrix $\widehat\Sigma=Z'Z/n$ converges weakly, with probability 1, to a limiting distribution. This limiting distribution is called the empirical spectral distribution (ESD). Let $Y$ be a random variable distributed according to the ESD of $\widehat\Sigma$, then the well-known Stieltjes transform $m\lam$ is
\begin{eqnarray*}
	m\lam=\E \left[ \frac{1}{Y+\lambda} \right].
\end{eqnarray*}
The Stieltjes transform of the spectral measure of $\widehat\Sigma$ satisfies
\begin{eqnarray*}
	\widehat{m}\lam=\frac{1}{k} \tr\big( (Z'Z/n+\lambda I_k)^{-1} \big) \to_{a.s.} m\lam \, .
\end{eqnarray*}

Additionally, we need the companion Stieltjes transform $v\lam$, which is the Stieltjes transform of the ESD of the $n \times n$ matrix $\underline{\widehat\Sigma}=ZZ'/n\,$. Let $\underline{Y}$ be a random variable distributed according to the ESD of $\underline{\widehat\Sigma}\, $, then
\begin{eqnarray*}
	v\lam=\E \left[ \frac{1}{\underline{Y}+\lambda} \right] ,
\end{eqnarray*}
and the Stieltjes transform of the spectral measure of $\underline{\widehat\Sigma}$ satisfies
\begin{eqnarray*}
	\widehat{v}\lam=\frac{1}{n} \tr\big( (ZZ'/n+\lambda I_n)^{-1} \big) \to_{a.s.} v(-\lambda) \, .
\end{eqnarray*}

To calculate the population quantities $m\lam$ and $v\lam$ as a benchmark in the simulations, we use the procedure described in \cite{dobriban_wager_supplement}. For estimation, we follow \cite{dobriban2020wonder} and use $\widehat{m}\lam$ and $\widehat{v}(-\lambda) \, $.

In the following, we repeatedly use well-known results in the RMT literature (the proofs are replicated in Appendix \ref{appendix_proof_lemmas} for convenience). First, we mention a useful representation of the Ridge smoother. 
\begin{lemma}\label{Lemma_ridgerep}
	\begin{align*}
		P_{\lambda} = I_n - \lambda (ZZ'/n+\lambda I_n)^{-1} \, .
	\end{align*}
\end{lemma}

Second, we mention a useful link between the Stieltjes transform and its companion. 
\begin{lemma}\label{Lemma_companion}
	\begin{align*}
		\lambda v\lam = 1 - \gamma (1-\lambda m\lam) \, .
	\end{align*}
\end{lemma}

Third, we mention limits involving the Stieltjes transforms that help us to work out the links to the existing literature on many weak instruments. 
\begin{lemma}\label{Lemma_vlam0}
	For the ridgeless limit ($\lambda\downarrow 0$), we have
	\begin{itemize}
		\item[] if $\gamma<1$, then \ 1a) $\lim_{\lambda\downarrow 0} \lambda v\lam = 1-\gamma$ \ and \ 1b) $\lim_{\lambda\downarrow 0} \lambda m\lam = 0 \, ,$ \vspace{-0.2cm}
		\item[] if $\gamma>1$, then \ 2a) $\lim_{\lambda\downarrow 0} \lambda v\lam = 0$ \ and \ 2b) $\lim_{\lambda\downarrow 0} \lambda m\lam = 1-\frac{1}{\gamma} \, .$
	\end{itemize}
\end{lemma}

We, moreover, make use of the Silverstein equation \citep{silverstein1995strong}, which is
\begin{align*}
	\frac{1}{v\lam} = \lambda + \gamma \int \frac{t}{1 + tv\lam} dH(t)
\end{align*}
where $\frac{dH(t)}{dt}$ is the pdf of the PSD. The Silverstein equation links the limit PSD to the limit ESD. If instruments are equi-correlated with correlation coefficient $\rho_z$, $\Sigma$ has $k-1$ eigenvalues that equal $1-\rho_z$, thus, the pdf of the PSD is a point mass at $1-\rho_z$. Therefore, the Silverstein equation has a closed-form expression
\begin{align*}
	\frac{1}{v\lam} = \lambda + \gamma \frac{1-\rho_z}{1 + ( 1-\rho_z ) v\lam} \, .
\end{align*}


In the special case of isotropic instruments ($\rho_z=0$), this simplifies to
\begin{align*}
	\frac{1}{v\lam}=\lambda + \frac{\gamma}{1+v\lam}
\end{align*}
\citep[see also][]{hastie2022supplement}.

\section{2SLS Estimation and RMT}

\subsection{Bias Approximation}
We derive a bias approximation for the causal effect $\beta$ when the first-stage \mbox{regression} is based on Ridge. Consider the first-stage Ridge estimates $\widehat\pi_\lambda$ from a regression of the endogenous variable $x$ on the instruments $Z$ given by
\begin{align*}
	\widehat\pi_\lambda = \argmin_{\pi} \left\{ \frac{1}{n}\Vert x-Z\pi \Vert^2_2 + \lambda \Vert \pi \Vert^2_2 \right\} \, ,
\end{align*}
or, equivalently written as $\widehat\pi_\lambda=\big( Z'Z/n+\lambda I_k \big)^{-1}Z'x/n\, $. In the second stage, we use the predicted values $\widehat{x}_{\lambda}=Z\widehat\pi_\lambda$ in a just-identified 2SLS regression to obtain
\begin{equation}
	\widetilde\beta_{\lambda} = \frac{x'P_{\lambda}y}{x'P_{\lambda}x } \, .
\end{equation}
In the following, we call $\widetilde\beta_\lambda$ the 2SLS-Ridge estimator to emphasize that it is a two-stage procedure with Ridge in the first stage and just-identified 2SLS in the second stage regression. For the limit $\lambda\downarrow 0$, we get the (ridgeless) 2SLS estimator, which corresponds to standard 2SLS if $\gamma<1$. As already mentioned in \cite{hansen2014instrumental}, 2SLS-Ridge is biased.  Theorem \ref{theorem_ridge} gives precise information about the bias induced by using Ridge in the first-stage regression. While the numerator only contains a noise term, the denominator is the sum of a signal term and a noise term. For convenience, we use the Stieltjes transform $m\lam$ for all terms related to the identification strength $\alpha^2$ and the companion Stieltjes transform $v\lam$ for all terms related to $\sigma^2_\nu$ or $\sigma_{\varepsilon\nu}$.
\begin{theorem}\label{theorem_ridge}
	Under Assumptions \ref{A_model}, \ref{A_hda}, and \ref{A_randcoeff}, suppose additionally that the eigenvalues of $\Sigma$ are uniformly bounded away from zero and infinity, and $\E[W_{ij}^8]$ are uniformly bounded from above. Then, the bias of 2SLS-Ridge estimation converges almost surely to
	\begin{align*}
		\widetilde\beta_{\lambda}-\beta \to_{a.s.} \frac{ \sigma_{\varepsilon\nu} \big( 1 - \lambda v \lam \big) }{ \alpha^2 \big(1-\lambda p\lam\big) + \sigma_\nu^2 (1-\lambda v\lam) } 
	\end{align*}
	with $p\lam\equiv 1-\lambda m\lam $.
\end{theorem}
\begin{proof}
	See Appendix \ref{appendix_proof_ridge}.
\end{proof}

Before discussing ways to get rid of the bias term later in Section \ref{Section_adjustment}, we briefly address the links of our bias approximation to the existing literature and the advantage of using Ridge compared to least-squares in the first-stage regression. The first part of the next corollary replicates the well-established 2SLS bias approximation using Theorem \ref{theorem_ridge} and Lemma \ref{Lemma_vlam0} (compare, for example, \cite{hahn2002new,hahn2002notes} and \cite{chao2005consistent}). The second part of the corollary extends these results to ridgeless 2SLS estimation. 
\begin{corollary}\label{corollary_biasridge}
	For the ridgeless limit ($\lambda\downarrow 0$), we have
	\begin{align*}
		\text{if $\gamma<1$, then} & \ \ \lim_{\lambda\downarrow 0} \frac{ \sigma_{\varepsilon\nu} \big( 1 - \lambda v \lam \big) }{ \alpha^2 \big(1-\lambda p\lam\big) + \sigma_\nu^2 (1-\lambda v\lam) }  = \frac{\gamma\sigma_{\varepsilon\nu}}{\alpha^2 + \gamma \sigma^2_\nu}= \frac{\sigma_{\varepsilon\nu}}{\sigma^2_\nu} \ \frac{1}{ \mu^2/k + 1} \, , \\
		\text{if $\gamma>1$, then} & \ \ \lim_{\lambda\downarrow 0} \frac{ \sigma_{\varepsilon\nu} \big( 1 - \lambda v \lam \big) }{ \alpha^2 \big(1-\lambda p\lam\big) + \sigma_\nu^2 (1-\lambda v\lam) }  = \frac{\sigma_{\varepsilon\nu}}{\alpha^2 + \sigma^2_\nu}= \frac{\sigma_{\varepsilon\nu}}{\sigma^2_\nu} \ \frac{1}{ \mu^2/n + 1} \, ,
	\end{align*}
using the definition of the concentration parameter $\mu^2\equiv n\alpha^2/\sigma_\nu^2\, $.
\end{corollary}
Comparing these results with the bias of OLS \citep{hahn2002notes}, which is
\begin{align*}
	\frac{ \sigma_{\varepsilon\nu}}{\sigma^2_\nu} \  \frac{1}{\frac{\mu^2}{n}+1} \, ,
\end{align*}
we find that the (ridgeless) 2SLS bias is never larger than the OLS bias (even if $\gamma>1$). The equivalence of the ridgeless 2SLS and OLS bias for $\gamma>1$ intuitively follows from the fact that the ridgeless 2SLS estimator interpolates the data in the first stage if $\gamma>1$, i.e. $x=Z\widehat\pi_{(\lambda\downarrow 0)}=\widehat{x}_{(\lambda\downarrow 0)}$. More interestingly, we can also show that increasing $\lambda$ reduces the bias relative to the (ridgeless) 2SLS bias (even if $\gamma>1$). To see this, we can rearrange the bias of 2SLS-Ridge (when we rule out the case $\lambda v\lam=1$, which occurs if and only if $\lambda\to\infty$)
\begin{align*}
	\frac{ \sigma_{\varepsilon\nu} \big( 1 - \lambda v \lam \big) }{ \alpha^2 \big(1-\lambda p\lam\big) + \sigma_\nu^2 (1-\lambda v\lam) } 
	=\frac{ \sigma_{\varepsilon\nu}}{\sigma^2_\nu} \  \frac{1}{F \, a\lam+1} \, .
\end{align*}
$F \equiv \mu^2/k$ is the population first-stage $F$-statistic and $a\lam\equiv \gamma\frac{1-\lambda p\lam}{1-\lambda v\lam}$ can be interpreted as a Ridge amplifier of the first-stage $F$-statistic. We show that this amplifier equals unity for the standard 2SLS (i.e., $\lambda=0$). Additionally, we show that its first derivative close to the ridgeless limit is positive, which suggests that there is a value of $\lambda$ for which the bias of 2SLS-Ridge is smaller compared to (ridgeless) 2SLS. 

\begin{corollary}\label{corollary_amplifier}
	If $E[Y^{-1}]\geq1$, $\gamma<1$ and $\lambda v\lam \neq 1$, then
	\begin{align*}
		&\lim_{\lambda\downarrow 0} a\lam = \lim_{\lambda\downarrow 0} \left(\frac{1}{p\lam}-\lambda\right) =1 \, , \\ 
		&\lim_{\lambda\downarrow 0} a'\lam = -\lim_{\lambda\downarrow 0} \left(\frac{p'\lam}{p^2\lam}+1\right) \geq0 \, ,
	\end{align*}
	with $-p'\lam=m\lam-\lambda m'\lam$ and $m'\lam=\E\left[ \frac{1}{(Y+\lambda)^2} \right] \, $.
\end{corollary}
\begin{proof}
	See Appendix \ref{appendix_proof_corollaries}.
\end{proof}

Figure \ref{fig_bias_2sls} depicts the bias approximation from Theorem \ref{theorem_ridge} and the simulated bias of $\widetilde\beta_\lambda$ for different values of $\lambda$. In the left part of Figure \ref{fig_bias_2sls}, we see the standard bias approximation for 2SLS if $\lambda\downarrow 0$ and $\gamma<1$, while in the right part we can observe that the ridgeless 2SLS bias equals the OLS bias if $\gamma>1$, which illustrates the results in Corollary \ref{corollary_biasridge}. Independent of $\gamma$, we see that increasing $\lambda$ reduces the bias of 2SLS-Ridge. However, while increasing $\lambda$ reduces the bias of $\widetilde\beta_\lambda$, we cannot expect to get rid of it fully. Therefore, we propose a bias-adjustment for the 2SLS-Ridge estimator in the following. 

\begin{figure}[h!]
	\centering
	\includegraphics[width=0.50\textwidth]{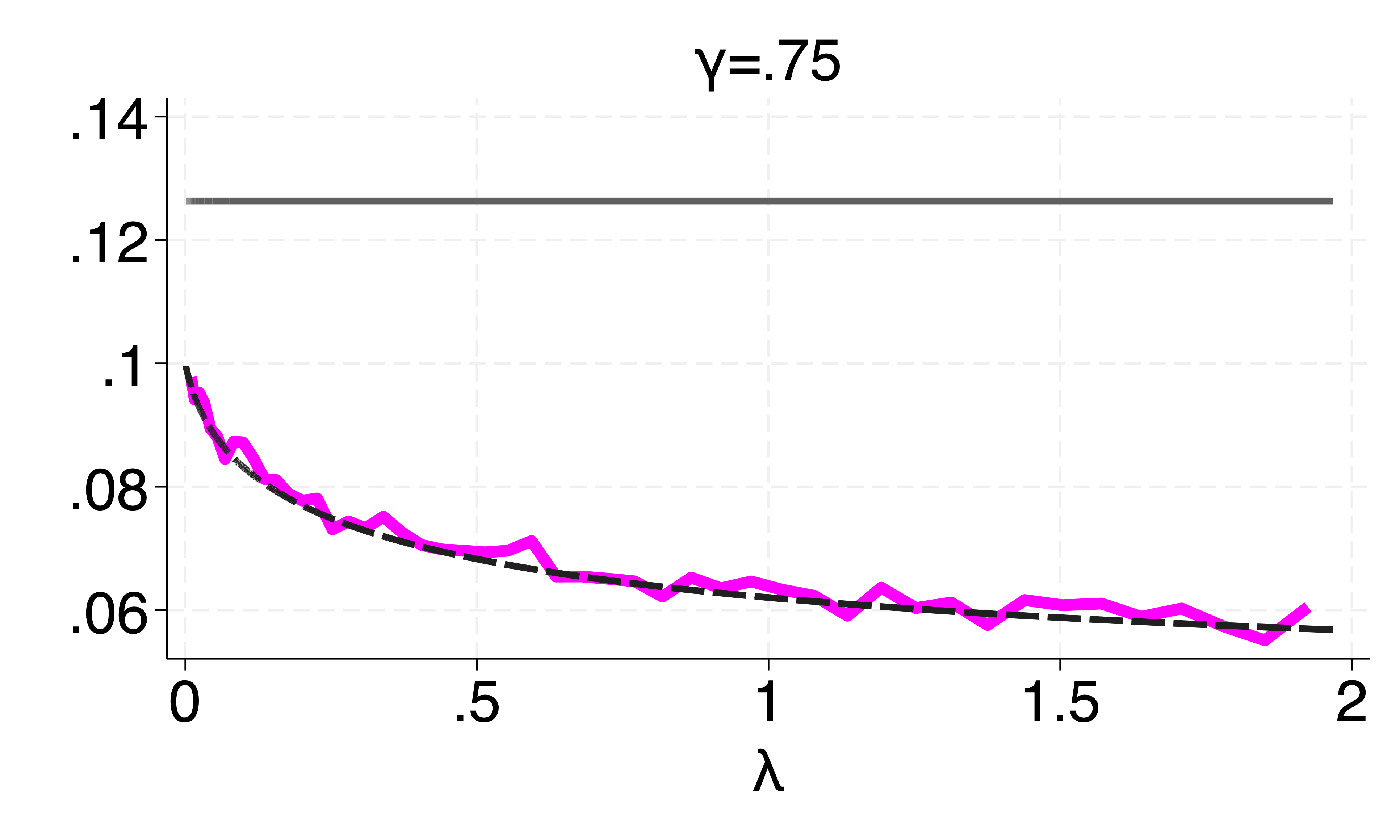}\hfill
	\includegraphics[width=0.50\textwidth]{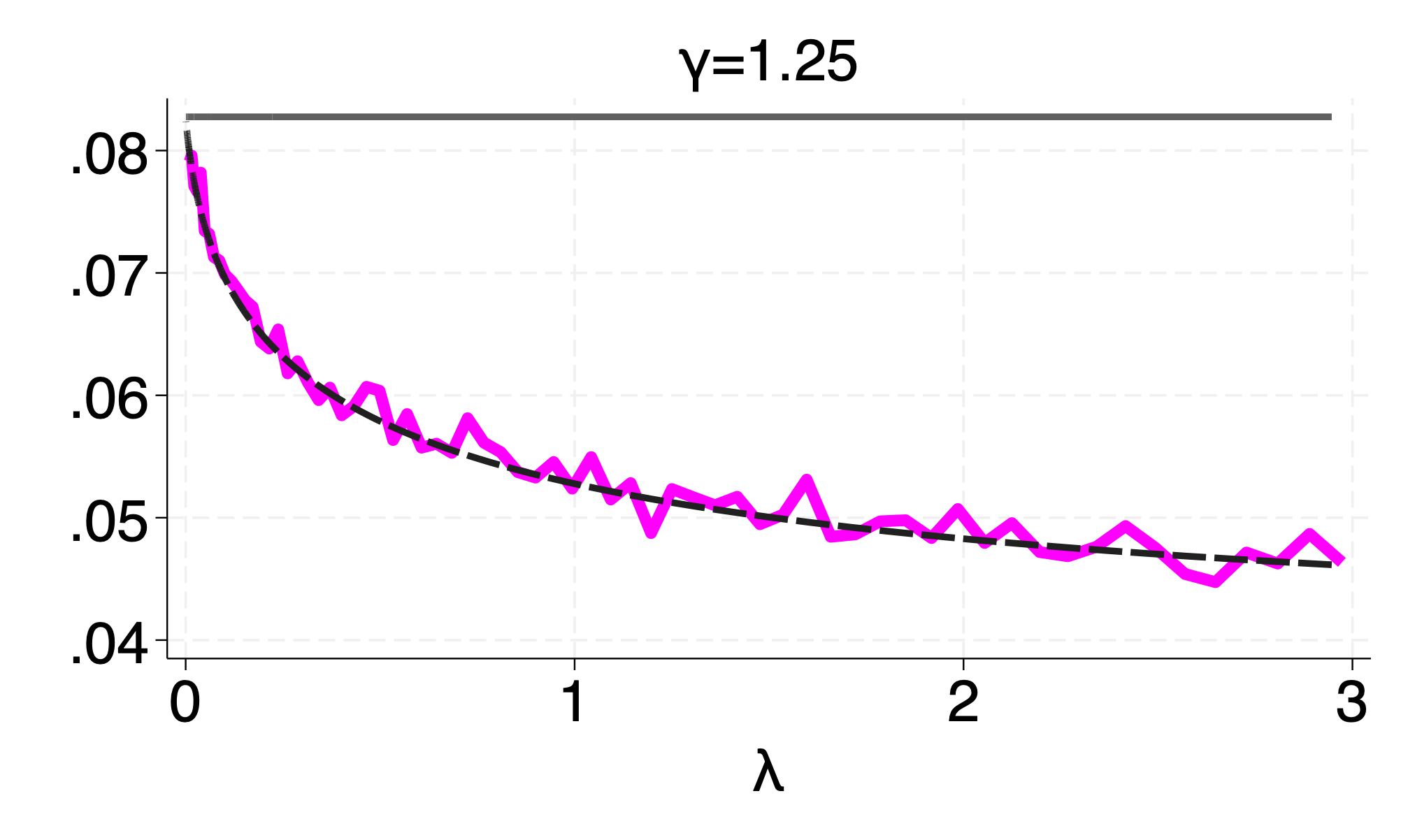}
	\caption{Bias approximation and simulation results for 2SLS-Ridge (average over 500 replications); dashed line shows RMT results from Theorem \ref{theorem_ridge}; solid grey line indicates OLS bias; AR-1 model $\Sigma_{ij}=0.5^{\vert i-j\vert}$. \label{fig_bias_2sls}}
\end{figure}

\subsection{Bias-Adjustment \label{Section_adjustment}}

The standard bias-adjusted 2SLS (aka Nagar's estimator, \citeyear{nagar1959bias}; \cite{donald2001choosing}; \cite{chao2005consistent}) has several representations
\begin{align*}
	\widehat\beta_0 &= \frac{x'P_0y - \frac{\gamma}{1-\gamma}x'M_0y}{x'P_0x - \frac{\gamma}{1-\gamma}x'M_0x} \\
	&= \frac{(1-\gamma)x'P_0y - \gamma x'M_0y}{(1-\gamma)x'P_0x - \gamma x'M_0x} \\
	&= \frac{x'P_0y - \gamma x'y}{x'P_0x - \gamma x'x} \, .
\end{align*}
Analogous to the bias-adjusted OLS projection matrix, $P_0 - \gamma I_n \, $, we propose a bias-adjusted Ridge smoother
\begin{align*}
	S_\lambda(s) = P_\lambda - s I_n \, .
\end{align*}
For the first term, we know $\tr(P_\lambda)/n\to_{a.s.} 1-\lambda v\lam$. Therefore, setting
\begin{align*}
	s_1 = \gamma \int \frac{t v\lam}{1 + tv\lam} dH(t) \, ,
\end{align*}
we know from the Silverstein equation that $\tr\big(S_\lambda(s_1)\big)/n \to_{a.s.} 0$. Unfortunately, the integral depends on unknown population values and cannot be derived in general. There are, however, several approaches to proceed. First, we may consider the case with $\lambda\downarrow 0$ and $\gamma<1$. Rearranging $s_1$ and using the dominated convergence theorem together with Lemma \ref{Lemma_vlam0}, we can derive that
\begin{align*}
	\lim_{\lambda\downarrow 0} s_1 = \gamma \lim_{\lambda\downarrow 0} \int \frac{t}{\frac{\lambda}{\lambda v\lam} + t} dH(t) = \gamma \, ,
\end{align*}
because $\lim_{\lambda\downarrow 0} \lambda v\lam = 1-\gamma$ and hence $\lim_{\lambda\downarrow 0} \lambda / \lambda v\lam = 0$. Again invoking the Silverstein equation gives us 
\begin{align*}
	\tr\big(S_0(s_1)\big)/n = \tr(P_0 - \gamma I_n)/n \to_{a.s.} 0 \, ,
\end{align*}
which, for $\gamma<1$, can be used to construct the standard bias-adjusted 2SLS. Second, the availability of a closed-form expression for the Silverstein equation can be used to debias the 2SLS-Ridge estimator for values of $\lambda$ different from zero and/or values of $\gamma$ larger than unity. We illustrate this for isotropic instruments ($\Sigma=I_k$). When we consider
\begin{align*}
	s_2 = \gamma\frac{v_I\lam}{1+v_I\lam}
\end{align*}
with the limit Stieltjes transform for such instruments depending only on $\gamma$ and $\lambda$
\begin{align*}
	\lambda v_I\lam = \frac{1-\gamma-\lambda+\sqrt{(1-\gamma-\lambda)^2+4\lambda}}{2} \, ,
\end{align*}
then we know from the Silverstein equation that $\tr\big( S_\lambda(s_2)\big)/n\to_{a.s.} 0$.

Third, we can use an estimate of $v\lam$ to construct an $s$ for general matrices $\Sigma$ that fulfill Assumption \ref{A_hda}. Consider
\begin{align*}
	s_3 = 1-\lambda \widehat{v}\lam
\end{align*}
with 
\begin{eqnarray*}
	\widehat{v}\lam=\frac{1}{n} \tr\big( (ZZ'/n+\lambda I_n)^{-1} \big) \to_{a.s.} v(-\lambda) \, , 
\end{eqnarray*}
and 
\begin{align*}
	S_{\lambda}(s_3)=P_{\lambda} - (1-\lambda \widehat{v}\lam)I_{n}
\end{align*}
such that by construction, $\tr\big(S_\lambda(s_3)\big)/n=0\, $. We focus on $S_\lambda(s_3)$, in the following simply referred to as $S_\lambda \, $, which gives us a bias-adjusted 2SLS-Ridge estimator that works for general values of $\lambda \, $,
\begin{equation}
	\widehat\beta_{\lambda} = \frac{x'S_{\lambda}y}{x'S_{\lambda}x } \, .
\end{equation}
The bias-adjustment utilizes that $P_\lambda$ acts differently on the signal and noise terms. By construction, the noise term is successfully eliminated while the signal is preserved. Theorem \ref{theorem_ba_ridge} shows that $\widehat\beta_{\lambda}$ estimates the causal effect consistently.
\begin{theorem}\label{theorem_ba_ridge}
	Under the conditions of Theorem \ref{theorem_ridge}, the bias of adjusted 2SLS-Ridge vanishes almost surely, i.e., 
	\begin{equation*}
		\widehat\beta_{\lambda}-\beta \to_{a.s.} 0 \, .
	\end{equation*}
\end{theorem}
\begin{proof}
	See Appendix \ref{appendix_proof_ba_ridge}.
\end{proof}

Using Lemma \ref{Lemma_vlam0}, we show in the following corollary that Theorem \ref{theorem_ba_ridge} also replicates the standard bias-adjusted 2SLS estimator as a special case (if $\gamma<1$ and $\lambda\downarrow 0$). Additionally, it extends this estimator to ridgeless 2SLS (if $\gamma>1$ and $\lambda\downarrow 0$).
\begin{corollary}\label{corollary_lambda0}
	For the ridgeless limit ($\lambda\downarrow 0$), we have
	\begin{align*}
		\text{if $\gamma<1$, then} \ \lim&_{\lambda\downarrow 0} \, S_\lambda = (1-\gamma)P_0 - \gamma M_0=P_0- \gamma I_n\, , \\
		\text{if $\gamma>1$, then} \ \lim&_{\lambda\downarrow 0} \, S_\lambda = - M_0 = P_0- I_n \, ,
	\end{align*}
with $P_0=Z(Z'Z)^{-1}Z'$ and $M_0=I_n-P_0 \, $.
\end{corollary}

Figure \ref{fig_bias_ba2sls} shows that $\widehat\beta_\lambda$ is (by construction) unbiased for all values of $\lambda$ (including $\lambda\downarrow 0$). The bias-adjustment works also if $\gamma>1\, $ (given that $\lambda\neq 0$). 

\begin{figure}[h!]
	\centering
	\includegraphics[width=0.50\textwidth]{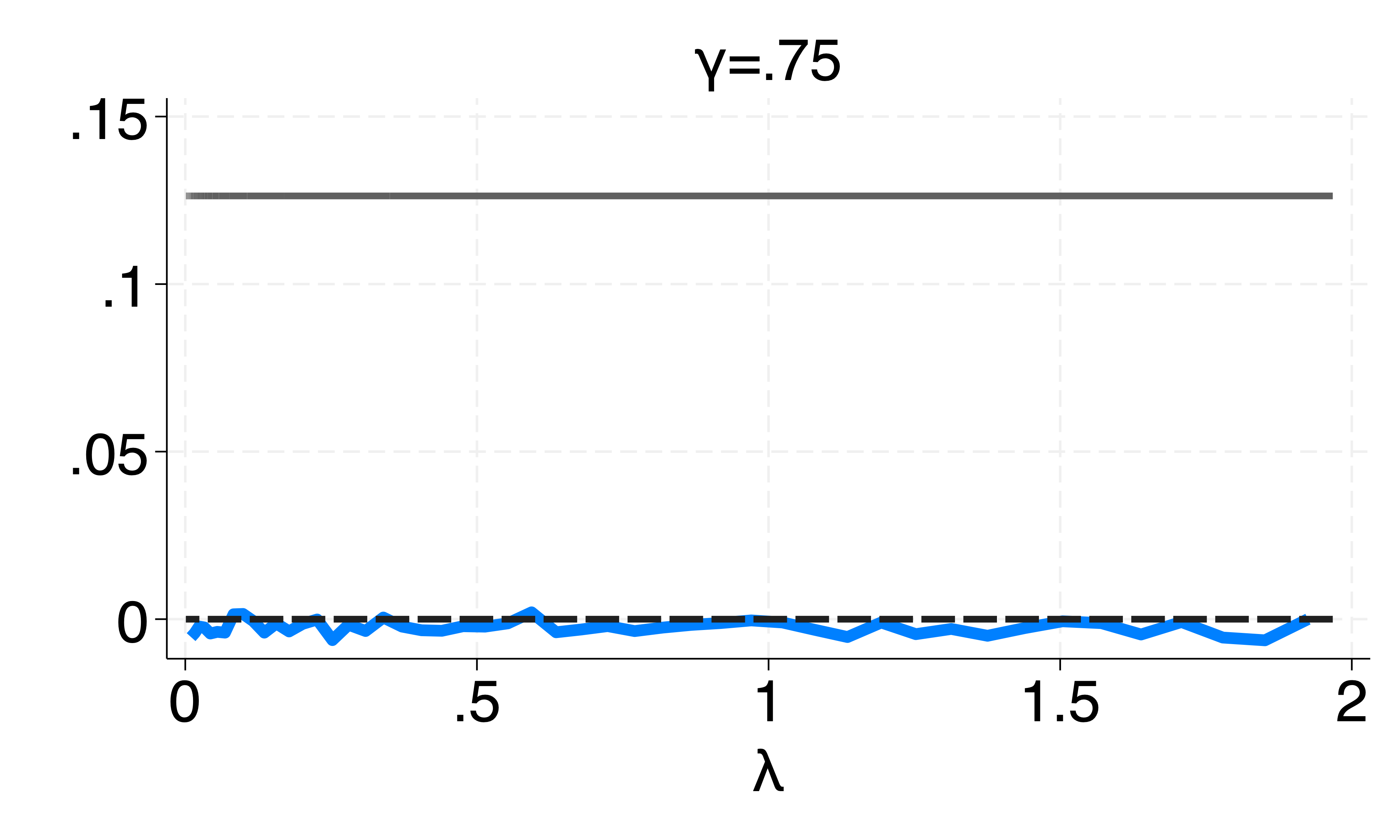}\hfill
	\includegraphics[width=0.50\textwidth]{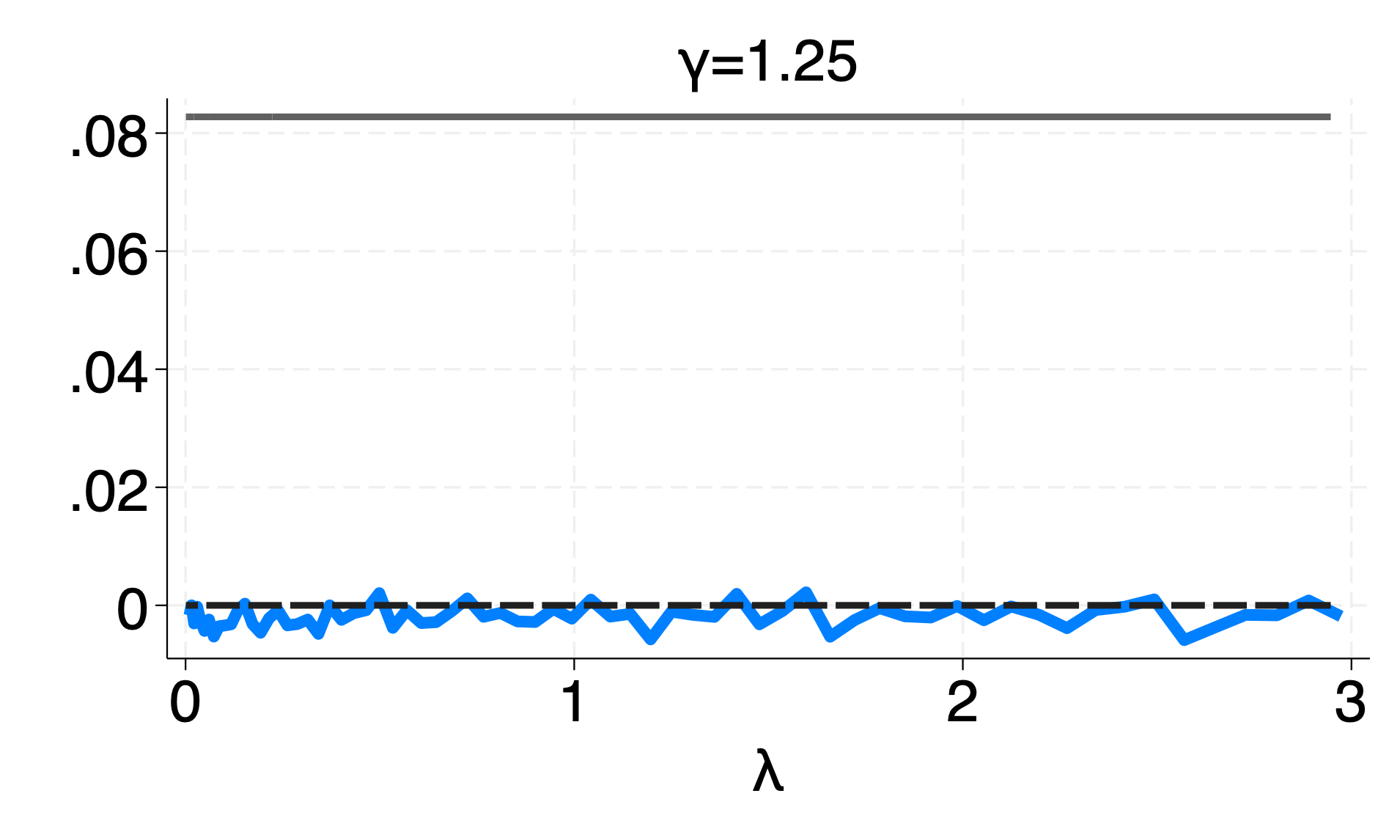}
	\caption{Bias approximation and simulation results for bias-adjusted 2SLS-Ridge (average over 500 replications); dashed line shows RMT results from Theorem \ref{theorem_ba_ridge}; solid grey line indicates OLS bias; AR-1 model $\Sigma_{ij}=0.5^{\vert i-j\vert}$. \label{fig_bias_ba2sls}}
\end{figure}

It is well known that the standard bias-adjusted 2SLS estimator (i.e., $\lambda=0$ and $\gamma<1$) can perform poorly in finite samples. A reason for this is that $\gamma$ in Corollary \ref{corollary_lambda0} can be interpreted as costs for the bias-adjustment, which reduce the signal to $\alpha^2(1-\gamma)$. By inspecting the limit of the denominator of bias-adjusted 2SLS-Ridge derived in the proof of Theorem \ref{theorem_ba_ridge} (see Appendix \ref{appendix_proof_ba_ridge}), we see two negative terms: 
\begin{align*}
	\frac{x'S_\lambda x}{n}  \to_{a.s.}  \alpha^2 \big( 1-(\lambda+\gamma) p\lam \big)\equiv \alpha^2 f\lam \, .
\end{align*}
The first term ($-\lambda p\lam$), which appears also in the asymptotic bias approximation of $\widetilde\beta_\lambda$ (see Theorem \ref{theorem_ridge}),  stems from using Ridge  and the second term ($-\gamma p\lam$) stems from the bias-adjustment. In both cases, $p\lam$ can be interpreted as a price per unit, which is unity if $\lambda=0$ and declines with increasing $\lambda$. This decline leads to a trade-off, allowing for less costly bias-adjustments (with respect to the signal term). 

The following corollary discusses this trade-off close to the ridgeless limit and shows under what conditions we can expect a stronger signal term if $\lambda$ increases. Compared with standard bias-adjusted 2SLS, the trade-off turns out to be positive if $\gamma \E\left[ Y^{-1} \right]\geq 1$. This condition simplifies to $\gamma>1/2 \, $ for isotropic instruments (i.e., when $\E\left[ T^{-1} \right]=1$) because from the Marchenko-Pastur equation, we know $\E(Y^{-1})=\E(T^{-1})/(1-\gamma)$.\footnote{In the next section, we additionally show that the asymptotic variance of $\widehat\beta_{\lambda}$ with $\lambda\neq 0$ turns out to be smaller than the asymptotic variance of the standard bias-adjusted 2SLS ($\lambda=0$) under the weaker condition $\E[Y^{-1}]\geq1$, which, for the isotropic case, is always fulfilled ($\gamma>0$).}

\begin{corollary}\label{corollary_badj_derivative}
	If $\gamma \E[Y^{-1}]\geq1$, $\gamma<1$ and $\lambda v\lam \neq 1$, then 
	\begin{align*}
		&\lim_{\lambda\downarrow 0} f\lam=\lim_{\lambda\downarrow 0} \big(1-(\lambda+\gamma) p\lam\big)=1-\gamma  \\
		&\lim_{\lambda\downarrow 0} f'\lam = -\lim_{\lambda\downarrow 0} \big((\lambda+\gamma)p'\lam + p\lam\big) \geq 0 \, .
	\end{align*}
\end{corollary}
\begin{proof}
	See Appendix \ref{appendix_proof_corollaries}.
\end{proof}

\begin{figure}[h!]
	\centering
	\includegraphics[width=0.50\textwidth]{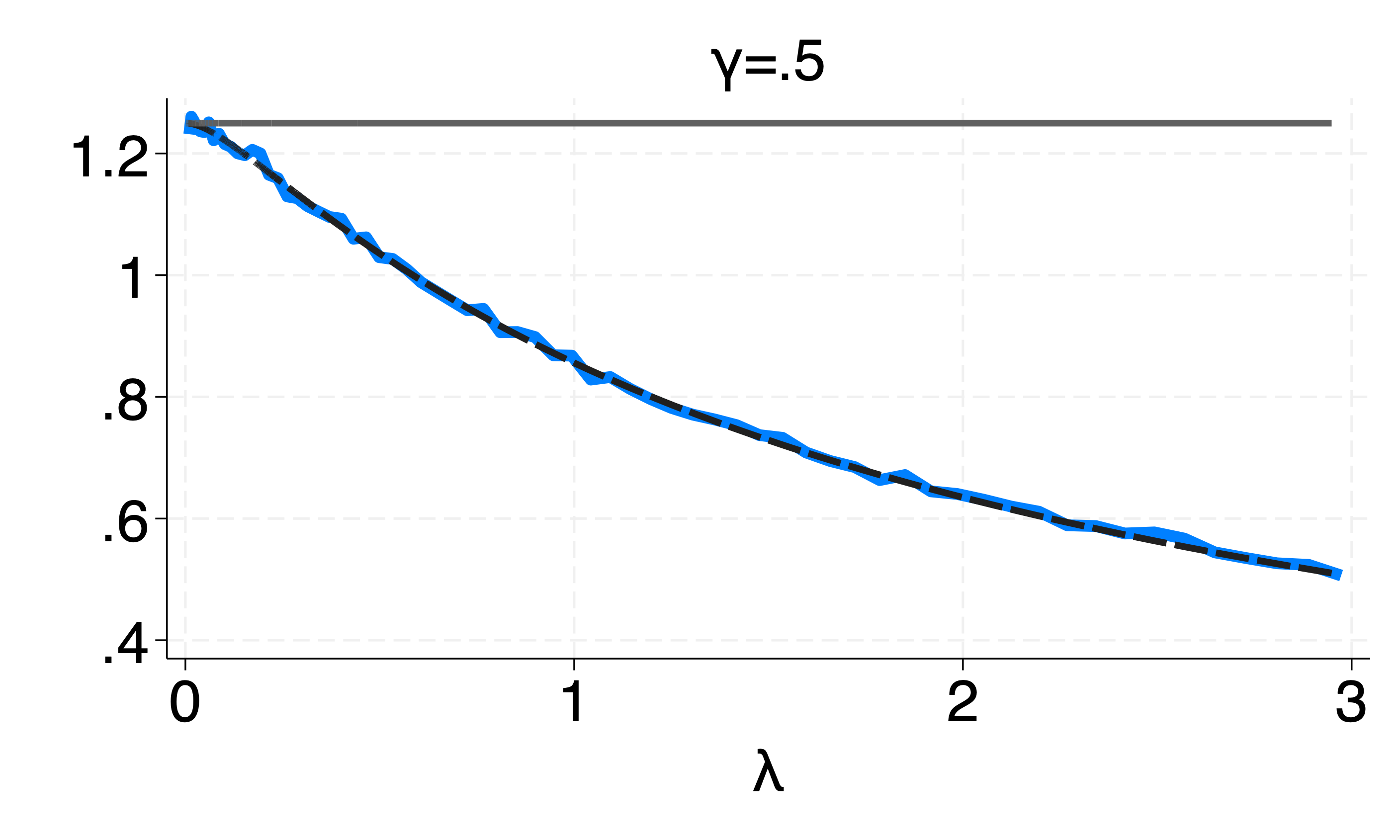}\hfill
	\includegraphics[width=0.50\textwidth]{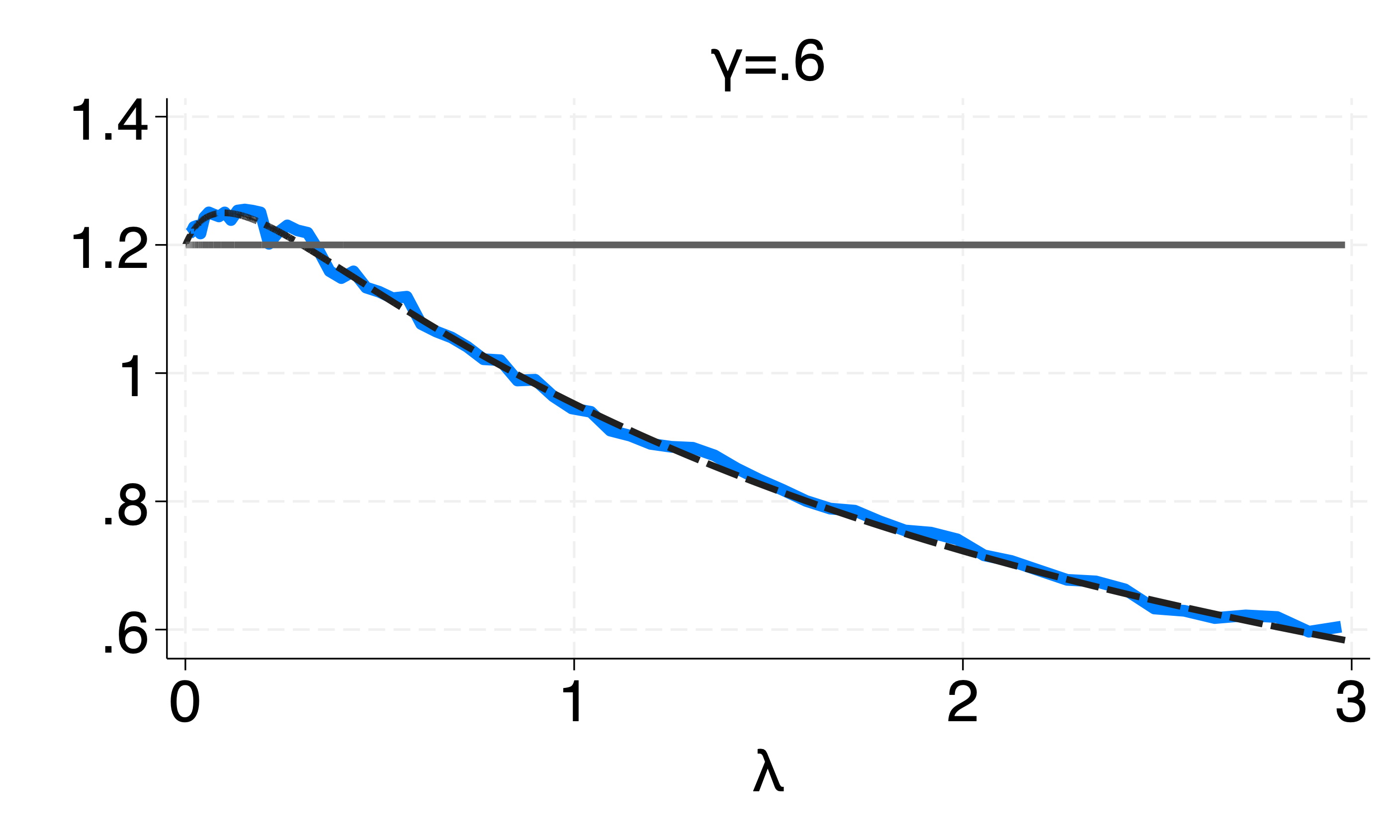}\hfill
	\includegraphics[width=0.50\textwidth]{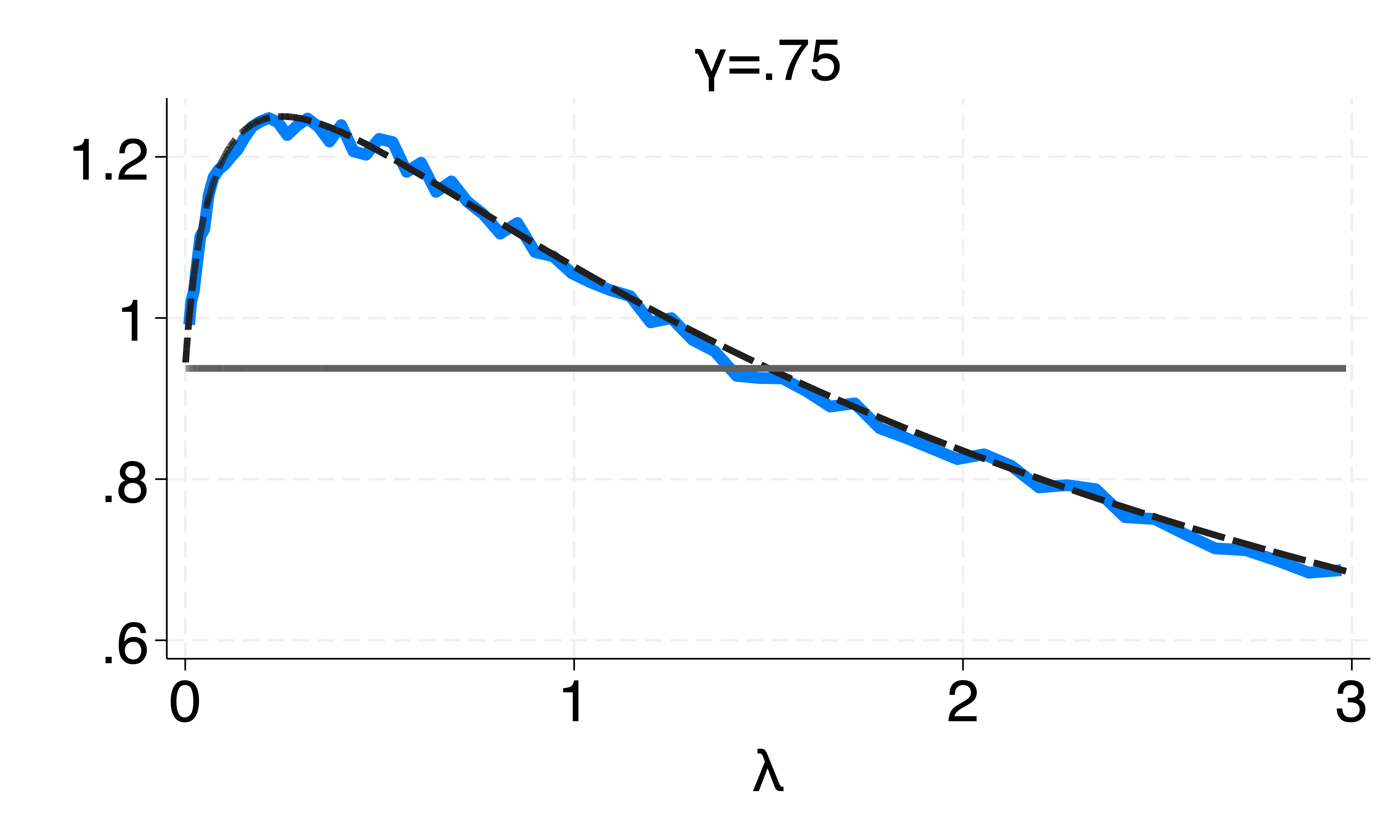}\hfill
	\includegraphics[width=0.50\textwidth]{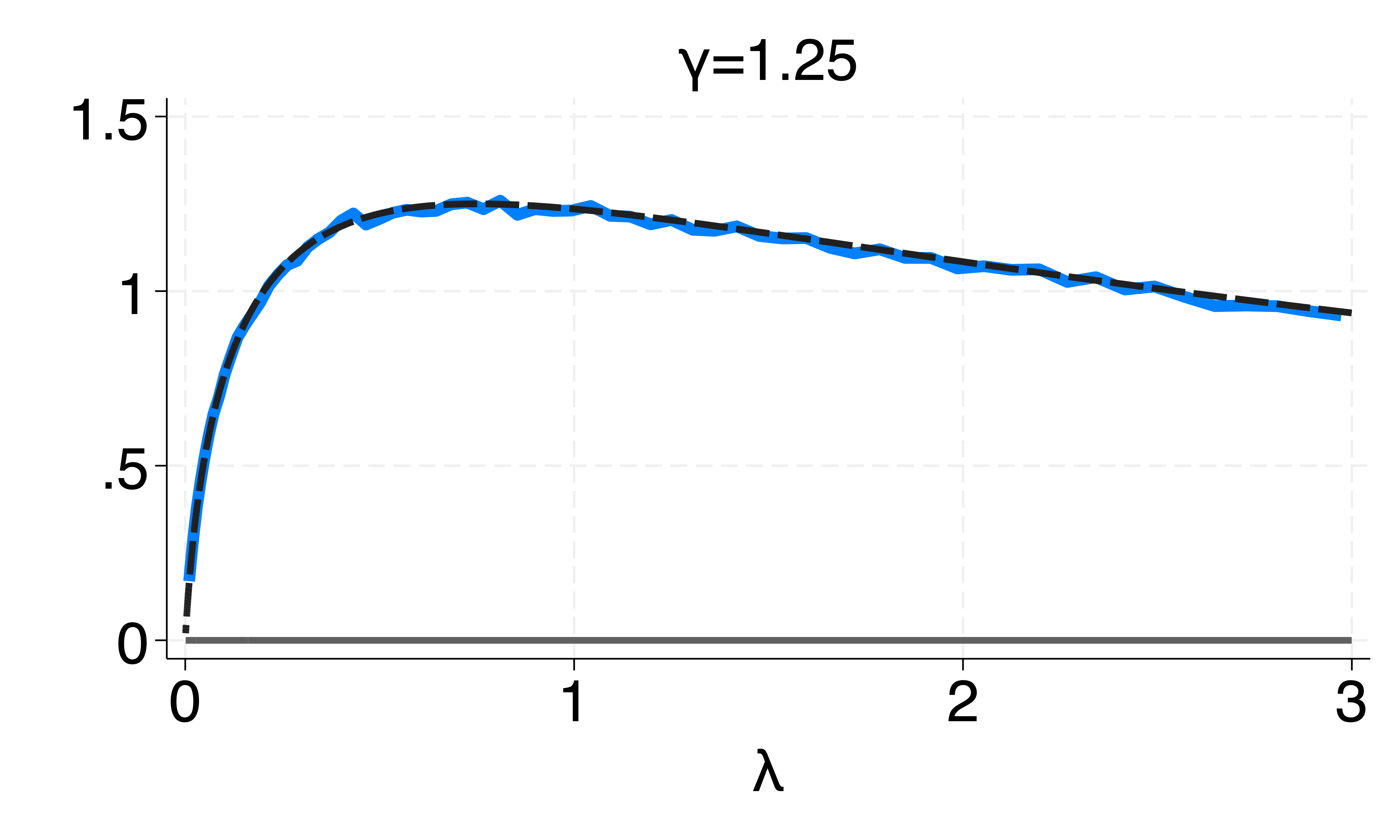}
	\caption{Signal approximation and simulation results for $x'S_\lambda x/n$ (average over 500 replications) and $\alpha^2 f\lam$ (dashed line); solid grey line indicates signal from standard bias-adjusted 2SLS, i.e. $\alpha^2(1-\gamma)$ if $\gamma<1$ and zero otherwise;  $\Sigma=I_k$. \label{fig_bias_signal}}
\end{figure}
Figure \ref{fig_bias_signal} illustrates the trade-off described in Corollary \ref{corollary_badj_derivative} for different values of $\gamma$ and $\lambda$. If $\gamma>1/2$, the signal can be improved compared to standard bias-adjusted 2SLS. If $\gamma>1$, we know for standard bias-adjusted 2SLS that $\lim_{\lambda\downarrow 0} f\lam = 0$ (using Lemma \ref{Lemma_vlam0}) implying that standard bias-adjusted 2SLS is infeasible. However, a positive value of $\lambda$ can reestablish the signal, which enables us to identify the causal effect even if $\gamma>1$. We observe similar patterns for the AR-1 covariance matrix (see Figure \ref{figappendix_bias_signal} in Appendix \ref{appendix_simulation}).

\section{Variance and Choice of the Tuning Parameter}

\subsection{Variance Estimation and Asymptotic Variance}

For inference, we consider \cite{bekker1994alternative} standard errors. Let $\widehat\varepsilon=y-x\widehat\beta_\lambda$, $\widehat\sigma_\varepsilon^2=\widehat\varepsilon'\widehat\varepsilon/n$, and $\widetilde{x} = x + \widehat\varepsilon \, (\widehat\varepsilon'x)/(\widehat\varepsilon'\widehat\varepsilon)$. Inspired by \cite{hansen2008estimation}, we can estimate the variance of the standard bias-adjusted 2SLS estimator by
\begin{align*}
	\widehat\V(\widehat\beta_0) = \widehat\sigma_\varepsilon^2 \, \frac{x'S_0S_0\widetilde{x}}{(x'S_0x)^2} \, .
\end{align*}
Following this structure, we propose to estimate the variance of the bias-adjusted 2SLS-Ridge estimator by

\begin{align*}
	\widehat\V(\widehat\beta_\lambda) = \widehat\sigma_\varepsilon^2 \, \frac{x'S_\lambda S_\lambda \widetilde{x}}{(x'S_\lambda x)^2} \, .
\end{align*}
The corresponding asymptotic variance of $\widehat\beta_\lambda$ is (see Appendix \ref{appendix_variance} for a derivation)
\begin{align}\label{eq:var}
	\widehat\V(\widehat\beta_\lambda) \to_{a.s.} \frac{1}{n}\frac{\sigma^2_\varepsilon \sigma_\nu^2}{\alpha^4} \, \left( \gamma F - \frac{ \gamma F \lambda^2 \big[ p' + p^2 \big]}{f^2} + \frac{ \lambda^2 \big[ v' - v^2 \big]}{f^2} \left( 1 + \rho_{\varepsilon\nu}^2 \right) \right) \, ,
\end{align}


which for $\lambda\downarrow 0$ coincides with the well-known asymptotic variance of the standard bias-adjusted 2SLS estimator under normality of the errors \citep[compare, e.g., ][]{bekker1994alternative, anatolyev2013instrumental,kolesar2015identification,anatolyev2019many}

\begin{align*}
	\V(\widehat\beta_0) &= \frac{1}{n}\frac{\sigma^2_\varepsilon \sigma_\nu^2}{\alpha^4} \, \left( \gamma F + \frac{\gamma}{1-\gamma} \left( 1 + \rho_{\varepsilon\nu}^2 \right) \right) \, .
\end{align*}

In Appendix \ref{appendix_proof_variance_standard_ba2SLS}, we show that for any positive value of $\lambda$, the asymptotic variance of $\widehat\beta_\lambda$ is always smaller than the asymptotic variance of $\widehat\beta_0$. In Figure \ref{fig_variance}, we present results from Monte-Carlo simulations over a grid of different $\lambda$ values. The results illustrate that the observed variation of $\widehat\beta_\lambda$, and the estimated Bekker-type variance are close to the asymptotic variance of the bias-adjusted 2SLS-Ridge estimator. To showcase convergence of the observed variance to the asymptotic variance, we use $n=1000$ (for results using $n=200$, see Figure \ref{figappendix_variance} in Appendix \ref{appendix_simulation}). For settings where $\gamma < 1$, we further display the asymptotic variances of the standard bias-adjusted 2SLS and the limited information maximum likelihood estimator (Liml) as benchmarks. As suggested by our asymptotic results, we see considerable efficiency improvements compared to the standard bias-adjusted 2SLS estimator, and, in some settings also to the Liml estimator. In Section \ref{section_lambda_star}, we further illustrate for which settings bias-adjusted 2SLS-Ridge outperforms Liml regarding estimation accuracy.

\begin{figure}[t!]
	\centering
	\includegraphics[width=0.50\textwidth]{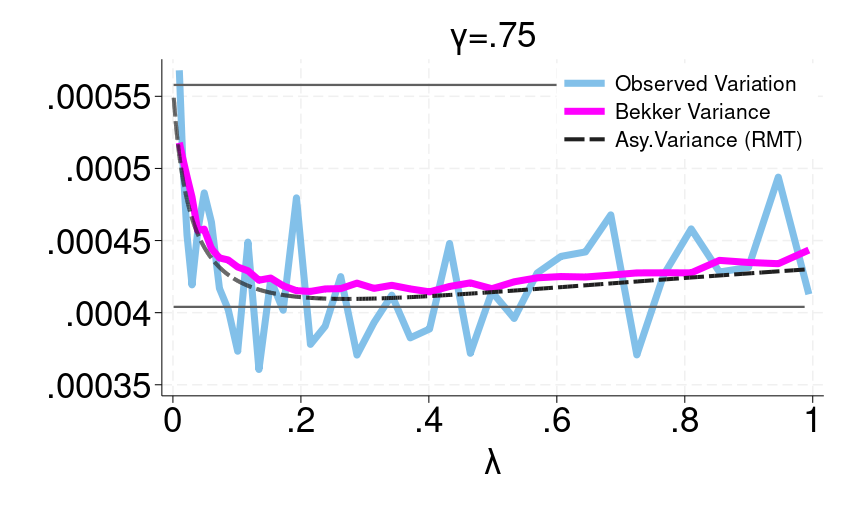}\hfill
	\includegraphics[width=0.50\textwidth]{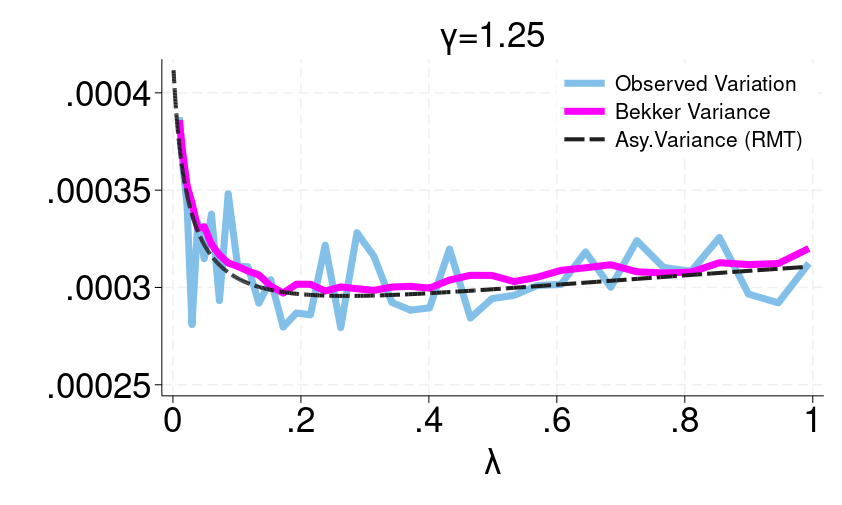}\hfill
	\includegraphics[width=0.50\textwidth]{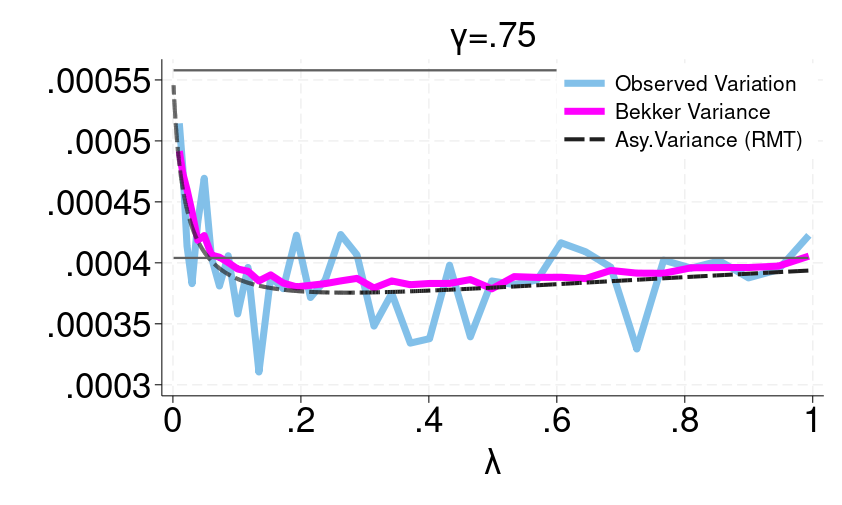}\hfill
	\includegraphics[width=0.50\textwidth]{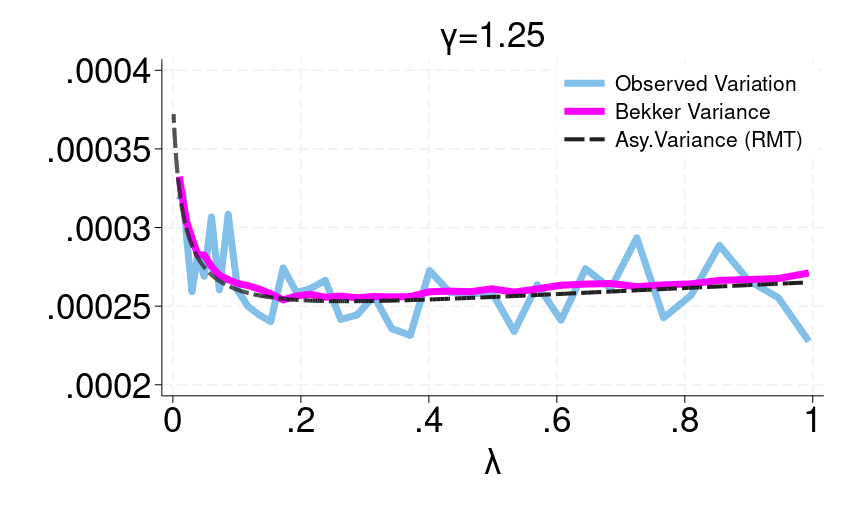}
	\caption{Observed variation and estimated variance of $\widehat\beta_\lambda$ (average over 500 replications) and asymptotic variance (dashed line) for $n=1000$. The upper and lower horizontal lines depict the asymptotic variance of the standard bias-adjusted 2SLS and Liml, respectively; top: $\Sigma=I_k$; bottom: AR-1 model $\Sigma_{ij}=0.5^{\vert i-j\vert}$. \label{fig_variance}}
\end{figure}

\bigskip

To further illustrate the validity of the variance estimates, Figure \ref{fig_rf} shows for different settings that under the null-hypothesis, observed rejection frequencies for testing at the 5\%-level are always close to their nominal value.  

\begin{figure}[h!]
	\centering
	\includegraphics[width=0.50\textwidth]{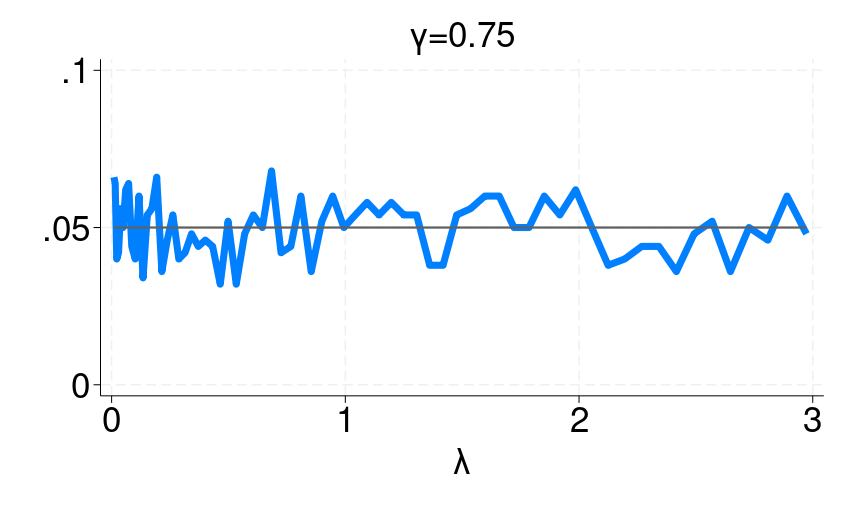}\hfill
	\includegraphics[width=0.50\textwidth]{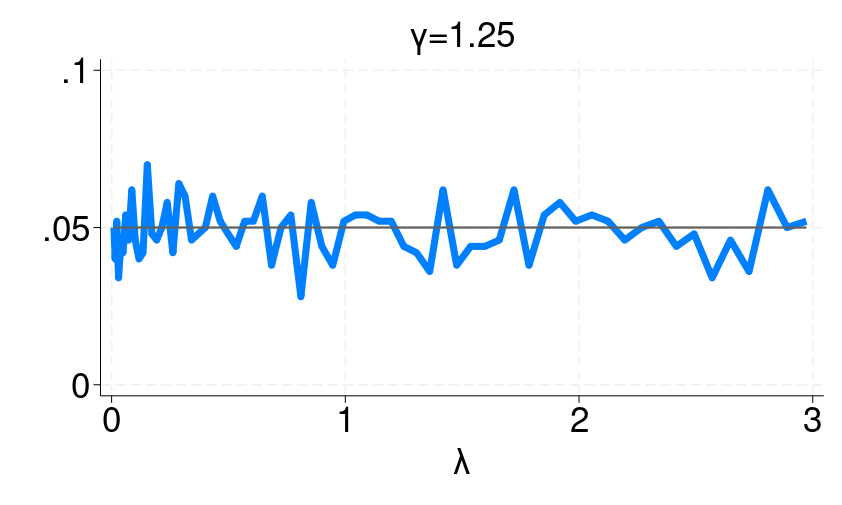}\hfill
	\includegraphics[width=0.50\textwidth]{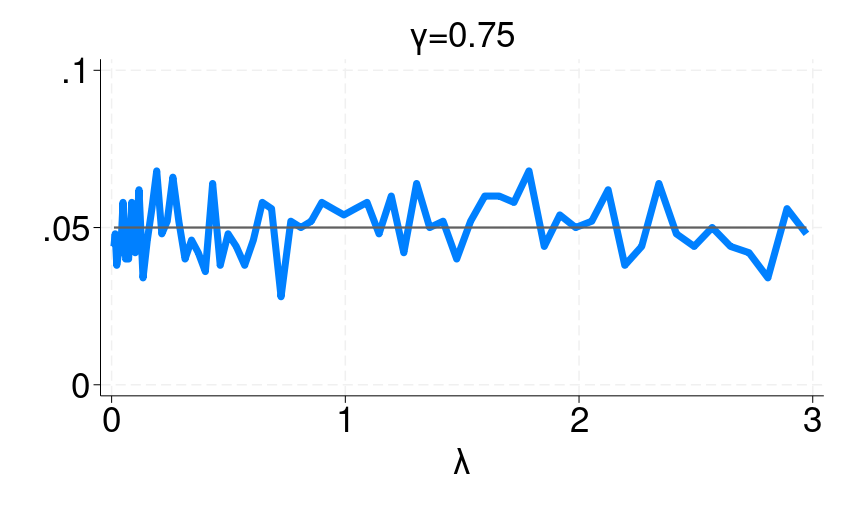}\hfill
	\includegraphics[width=0.50\textwidth]{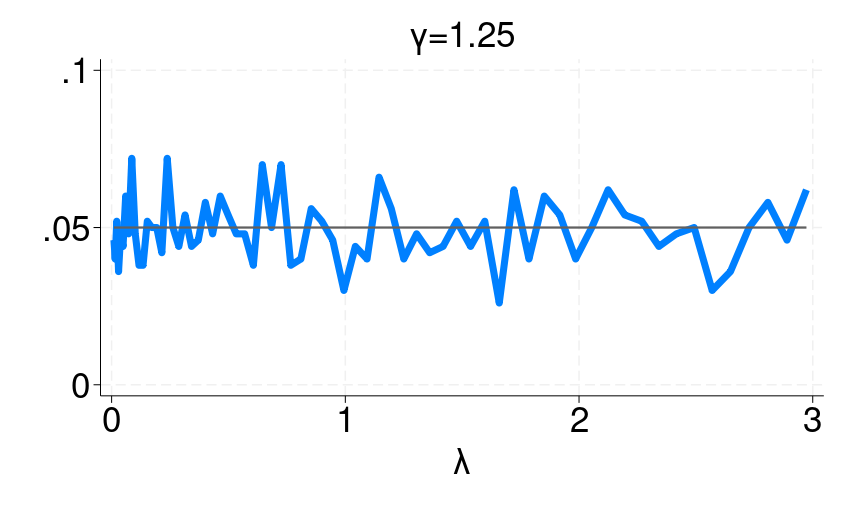}
	\caption{Rejection frequencies (average over 500 replications) for a t-test with null-hypothesis $\beta = 0$ for $n=1000$. top: $\Sigma=I_k$; bottom: AR-1 model $\Sigma_{ij}=0.5^{\vert i-j\vert}$. \label{fig_rf}}
\end{figure}

\bigskip


\subsection{Optimal Choice of the Tuning Parameter \label{section_lambda_star}}
Since $\widehat\beta_\lambda$ is consistent for all values of $\lambda$, we aim to find the optimal value of $\lambda$ that minimizes the variance of $\widehat\beta_\lambda$. To find this value theoretically, we take the derivative of the asymptotic variance (Equation \ref{eq:var}) and set it to zero. This gives us the following expression (for a complete derivation, see Appendix \ref{appendix_proof_lambda_star}):
\begin{equation}
    \begin{split}
           \frac{1}{n}\frac{\sigma^2_{\epsilon} \sigma^2_\nu F }{\alpha^4} & \bigg( \lambda - \frac{1+\rho_{\epsilon \nu}^2}{F} \bigg) * \\ & \bigg[ \frac{4 \frac{\lambda^3}{\gamma} vv' - 2 \frac{\lambda^3}{\gamma} v^3 - \frac{\lambda^3}{\gamma} v'' + \lambda^3 \left(1 + \frac{\lambda}{\gamma}\right) vv'' - 2 \lambda^3 \left(1 + \frac{\lambda}{\gamma}\right) v'^2 }{f^3} \bigg] = 0
    \end{split}
\end{equation}

This expression is fulfilled for $\lambda^* = \frac{1+\rho_{\varepsilon\nu}^2}{F}$. Interestingly, for $\rho_{\epsilon \nu}=0$, this result coincides with the optimal $\lambda^*$ in a single equation model which is equal to $\frac{1}{F}$ \citep{dobriban_wager}.

To obtain $\lambda^*$ empirically, we estimate $\rho_{\varepsilon\nu}$ using the formula for the OLS bias \citep{hahn2002notes}:
\begin{align}\label{eq:rhoevhat}
    \hat{\rho}_{\epsilon \nu} = \left( \hat{\beta}_{ols} - \hat{\beta}_{\downarrow 0} \right) \left[ (\hat{F} - 1) * \gamma + 1 \right]
\end{align}
with bias-adjusted 2SLS ($\hat{\beta}_{\downarrow 0}$) as consistent but inefficient preliminary estimate.

Finally, we can compute the optimal $\lambda$ using:

\begin{align}\label{eq:lhat}
    \hat{\lambda}^* = \frac{1+\hat{\rho}_{\epsilon \nu}^2}{\hat{F}} 
\end{align}

Figure \ref{fig_lambda_star} in Appendix \ref{appendix_simulation} illustrates for different settings that estimating $\hat\lambda^*$ via Equations \ref{eq:rhoevhat} and \ref{eq:lhat} gives a reliable estimate of $\lambda^* \, $.

\begin{figure}[b!]
	\centering
	\includegraphics[width=\textwidth]{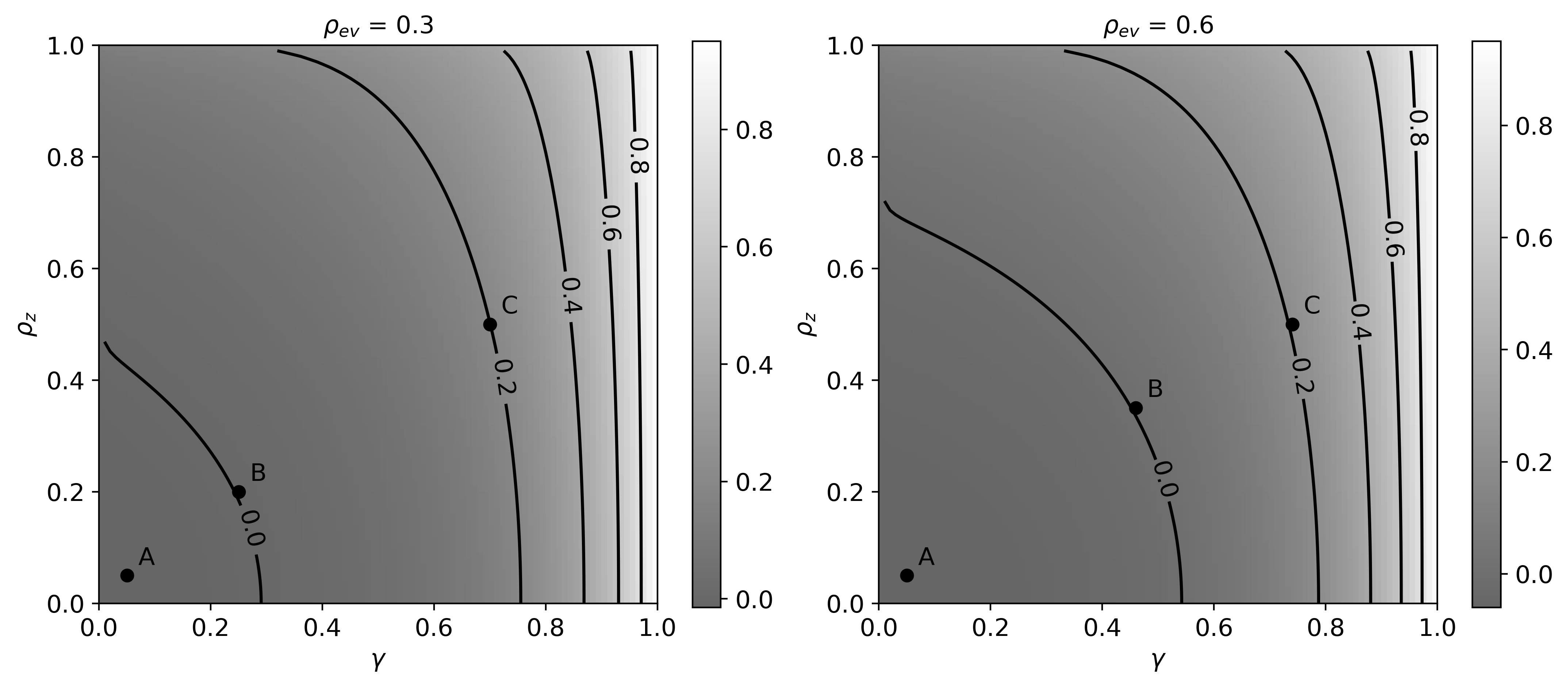}
	\caption{Contour plots of the relative difference in the asymptotic variance of Liml vs. bias-adjusted 2SLS-Ridge; positive values indicate that bias-adjusted 2SLS-Ridge is more efficient than Liml; AR-1 model; $F=5$; Black dots mark settings that are shown in Table \ref{tab_mc}. \label{fig_contour}}
\end{figure}

Next, we compare the asymptotic variance formula of bias-adjusted 2SLS-Ridge using the optimal tuning parameter $\lambda^*$ (see Appendix \ref{appendix_proof_lambda_star} for a derivation) to that of Liml:
\begin{align*}
	\V(\widehat\beta_\lambda) = \frac{1}{n}\frac{\sigma^2_\varepsilon \sigma_\nu^2}{\alpha^4} \, \left( \gamma F - \frac{ \gamma F \lambda^2 \big[ p' + p^2 \big]}{f^2} + \frac{ \lambda^2 \big[ v' - v^2 \big]}{f^2} \left( 1 + \rho_{\varepsilon\nu}^2 \right) \right) \, ,
\end{align*}
\begin{align*}
	\V(\widehat\beta_{\text{Liml}}) &= \frac{1}{n}\frac{\sigma^2_\varepsilon \sigma_\nu^2}{\alpha^4} \, \left( \gamma F - \frac{\gamma}{1-\gamma} \left( 1 + \rho_{\varepsilon\nu}^2 \right) \right) \, .
\end{align*}

In Figure \ref{fig_contour}, we plot the relative difference of the asymptotic variances of the two estimators over a grid of $\gamma$ and $\rho_z$ values for two distinct values of $\rho_{\epsilon \nu}$. We can see that the improvement in efficiency over Liml can be substantial, especially when the number of instruments approaches the number of observations.

\begin{table}[b!]
	\centering
	\begin{tabular}{c||l|rrr|rrr}
	\multicolumn{1}{c}{} & & \multicolumn{3}{c}{$\rho_{\epsilon \nu} = 0.3$} & \multicolumn{3}{|c}{$\rho_{\epsilon \nu} = 0.6$}\\	\hline \hline
	 \multicolumn{1}{c}{} & & \textbf{Mn.B.} & \textbf{emp.SD} & \textbf{IQR} & \textbf{Mn.B.} & \textbf{emp.SD}  & \textbf{IQR} \\
	\hline
	\multirow{5}{*}{A} & OLS & 0.23919 & 0.02896 & 0.03886 & 0.48079 & 0.03042 & 0.03951\\
	 & 2SLS & 0.04736 & 0.05788 & 0.08017 & 0.09917 & 0.05807 & 0.07177\\
	 & BA2SLS & -0.00593 & 0.07249 & 0.10037 & -0.00790 & 0.07748 & 0.09600\\
	 & Liml & -0.00367 & 0.07017 & 0.09774 & -0.00393 & 0.07180 & 0.08951\\
	 & BA2SLS-R & -0.00604 & 0.07230 & 0.10160 & -0.00767 & 0.07735 & 0.09541\\
	\hline
	\multirow{5}{*}{B} & OLS & 0.13385 & 0.02223 & 0.02967 & 0.18240 & 0.01939 & 0.02570\\
	 & 2SLS & 0.05041 & 0.02650 & 0.03343 & 0.10061 & 0.01976 & 0.02596\\
	 & BA2SLS & -0.00012 & 0.03319 & 0.04274 & 0.00013 & 0.02660 & 0.03507\\
	 & Liml & 0.00006 & 0.03281 & 0.04139 & -0.00038 & 0.02346 & 0.03204\\
	 & BA2SLS-R & -0.00005 & 0.03274 & 0.04196 & -0.00014 & 0.02510 & 0.03304\\
	\hline
	\multirow{5}{*}{C} & OLS & 0.06722 & 0.01526 & 0.01987 & 0.12809 & 0.01644 & 0.02171\\
	 & 2SLS & 0.05054 & 0.01570 & 0.02160 & 0.10063 & 0.01641 & 0.02142\\
	 & BA2SLS & 0.00029 & 0.02295 & 0.02958 & 0.00100 & 0.02454 & 0.03242\\
	 & Liml & 0.00039 & 0.02212 & 0.02842 & 0.00063 & 0.02102 & 0.02837\\
	 & BA2SLS-R & 0.00057 & 0.02002 & 0.02603 & 0.00041 & 0.02042 & 0.02721\\
\end{tabular}
	\caption{Results from Monte-Carlo simulations with 1000 replications for $n=1000$, $\beta = 0$, $F=5$; in settings A, Liml is expected to be more efficient than bias-adjusted 2SLS-Ridge, in B, both estimators should show equal performance, and in C, bias-adjusted 2SLS-Ridge should outperform Liml with respect to estimation accuracy; Columns display the following quality measures: mean bias, empirical standard-deviation, interquartile range.} 
	\label{tab_mc}	
\end{table}

Given Figure \ref{fig_contour}, we construct three settings for two different values of $\rho_{\epsilon \nu}$ each. The settings correspond to different relative performances of the two estimators. In setting A, our estimator is expected to perform worse than Liml with respect to estimation accuracy. Both estimators are expected to show similar performance in setting B, whereas in setting C, bias-adjusted 2SLS-Ridge should outperform Liml. In Table \ref{tab_mc}, we report estimation results for 2SLS, standard bias-adjusted 2SLS, Liml, and bias-adjusted 2SLS-Ridge with $\hat{\lambda}^*$ for each of these settings. While all bias-adjusted estimators show little bias, Liml and bias-adjusted 2SLS-Ridge outperform the standard bias-adjusted 2SLS estimator in terms of accuracy in every setting. The relative difference in observed variation between Liml and bias-adjusted 2SLS-Ridge reflects the asymptotic results, although bias-adjusted 2SLS-Ridge performs slightly worse than expected for the case when $\rho_{\epsilon \nu} = 0.6$.


\section{Conclusion \label{Section_conclusion}}

Instrumental variables estimation performs poorly when the number of instruments increases. In these settings, standard IV estimators are biased towards the OLS estimate. Traditional bias-adjustments can be motivated by the Silverstein equation. This link to random matrix theory allows us to improve IV estimation even in settings where no closed-form expression of the Silverstein equation exists. Bias-adjustments are costly in the sense that they reduce the overall signal strength of the instruments. We show that Ridge estimation of the first-stage parameters reduces the implicit price of bias-adjustments. This leads to a trade-off, allowing for less costly estimation of the causal effect which comes along with improved asymptotic properties. Our theoretical results nest existing methods for bias approximation and adjustment with OLS in the first-stage regression. Moreover, the results generalize to settings where the number of instruments is larger than the number of observations, which has not been discussed in the IV literature so far. As a byproduct, we derive the optimal tuning parameter of Ridge regressions in simultaneous equations models, which comprises the well-known results for single equation models as special case with uncorrelated error terms. 

\clearpage

\bibliographystyle{chicago} 
\bibliography{literature}

\begin{thebibliography}{}

\bibitem[\protect\citeauthoryear{Anatolyev}{Anatolyev}{2013}]{anatolyev2013instrumental}
Anatolyev, S. (2013).
\newblock Instrumental variables estimation and inference in the presence of
  many exogenous regressors.
\newblock {\em Econometrics Journal\/}~{\em 16\/}(1), 27--72.

\bibitem[\protect\citeauthoryear{Anatolyev}{Anatolyev}{2019}]{anatolyev2019many}
Anatolyev, S. (2019).
\newblock Many instruments and/or regressors: A friendly guide.
\newblock {\em Journal of Economic Surveys\/}~{\em 33\/}(2), 689--726.

\bibitem[\protect\citeauthoryear{Angrist and Frandsen}{Angrist and
  Frandsen}{2022}]{angrist2022machine}
Angrist, J.~D. and B.~Frandsen (2022).
\newblock Machine labor.
\newblock {\em Journal of Labor Economics\/}~{\em 40\/}(S1), S97--S140.

\bibitem[\protect\citeauthoryear{Bai and Silverstein}{Bai and
  Silverstein}{2010}]{bai2010spectral}
Bai, Z. and J.~W. Silverstein (2010).
\newblock {\em Spectral analysis of large dimensional random matrices}.
\newblock Springer.

\bibitem[\protect\citeauthoryear{Bekker}{Bekker}{1994}]{bekker1994alternative}
Bekker, P.~A. (1994).
\newblock Alternative approximations to the distributions of instrumental
  variable estimators.
\newblock {\em Econometrica\/}, 657--681.

\bibitem[\protect\citeauthoryear{Belloni, Chen, Chernozhukov, and
  Hansen}{Belloni et~al.}{2012}]{belloni2012sparse}
Belloni, A., D.~Chen, V.~Chernozhukov, and C.~Hansen (2012).
\newblock Sparse models and methods for optimal instruments with an application
  to eminent domain.
\newblock {\em Econometrica\/}~{\em 80\/}(6), 2369--2429.

\bibitem[\protect\citeauthoryear{Bound, Jaeger, and Baker}{Bound
  et~al.}{1995}]{bound1995problems}
Bound, J., D.~A. Jaeger, and R.~M. Baker (1995).
\newblock Problems with instrumental variables estimation when the correlation
  between the instruments and the endogenous explanatory variable is weak.
\newblock {\em Journal of the American Statistical Association\/}~{\em
  90\/}(430), 443--450.

\bibitem[\protect\citeauthoryear{Bun and Windmeijer}{Bun and
  Windmeijer}{2011}]{bun2011comparison}
Bun, M.~J. and F.~Windmeijer (2011).
\newblock A comparison of bias approximations for the two-stage least squares
  ({2SLS}) estimator.
\newblock {\em Economics Letters\/}~{\em 113\/}(1), 76--79.

\bibitem[\protect\citeauthoryear{Chamberlain and Imbens}{Chamberlain and
  Imbens}{2004}]{chamberlain2004random}
Chamberlain, G. and G.~Imbens (2004).
\newblock Random effects estimators with many instrumental variables.
\newblock {\em Econometrica\/}~{\em 72\/}(1), 295--306.

\bibitem[\protect\citeauthoryear{Chao and Swanson}{Chao and
  Swanson}{2005}]{chao2005consistent}
Chao, J.~C. and N.~R. Swanson (2005).
\newblock Consistent estimation with a large number of weak instruments.
\newblock {\em Econometrica\/}~{\em 73\/}(5), 1673--1692.

\bibitem[\protect\citeauthoryear{Couillet and Liao}{Couillet and
  Liao}{2022}]{couillet2022random}
Couillet, R. and Z.~Liao (2022).
\newblock {\em Random matrix methods for machine learning}.
\newblock Cambridge University Press.

\bibitem[\protect\citeauthoryear{Dobriban and Sheng}{Dobriban and
  Sheng}{2020}]{dobriban2020wonder}
Dobriban, E. and Y.~Sheng (2020).
\newblock {WONDER}: Weighted one-shot distributed {R}idge regression in high
  dimensions.
\newblock {\em Journal of Machine Learning Research\/}~{\em 21\/}(66), 1--52.

\bibitem[\protect\citeauthoryear{Dobriban and Wager}{Dobriban and
  Wager}{2018a}]{dobriban_wager}
Dobriban, E. and S.~Wager (2018a).
\newblock High-dimensional asymptotics of prediction: Ridge regression and
  classification.
\newblock {\em Annals of Statistics\/}~{\em 46\/}(1), 247--279.

\bibitem[\protect\citeauthoryear{Dobriban and Wager}{Dobriban and
  Wager}{2018b}]{dobriban_wager_supplement}
Dobriban, E. and S.~Wager (2018b).
\newblock Supplement to ``{H}igh-dimensional asymptotics of prediction: Ridge
  regression and classification''.
\newblock {\em Annals of Statistics\/}~{\em 46\/}(1), 247--279.

\bibitem[\protect\citeauthoryear{Donald and Newey}{Donald and
  Newey}{2001}]{donald2001choosing}
Donald, S.~G. and W.~K. Newey (2001).
\newblock Choosing the number of instruments.
\newblock {\em Econometrica\/}~{\em 69\/}(5), 1161--1191.

\bibitem[\protect\citeauthoryear{Gold, Lederer, and Tao}{Gold
  et~al.}{2020}]{gold2020inference}
Gold, D., J.~Lederer, and J.~Tao (2020).
\newblock Inference for high-dimensional instrumental variables regression.
\newblock {\em Journal of Econometrics\/}~{\em 217\/}(1), 79--111.

\bibitem[\protect\citeauthoryear{Hahn and Hausman}{Hahn and
  Hausman}{2002a}]{hahn2002new}
Hahn, J. and J.~Hausman (2002a).
\newblock A new specification test for the validity of instrumental variables.
\newblock {\em Econometrica\/}~{\em 70\/}(1), 163--189.

\bibitem[\protect\citeauthoryear{Hahn and Hausman}{Hahn and
  Hausman}{2002b}]{hahn2002notes}
Hahn, J. and J.~Hausman (2002b).
\newblock Notes on bias in estimators for simultaneous equation models.
\newblock {\em Economics Letters\/}~{\em 75\/}(2), 237--241.

\bibitem[\protect\citeauthoryear{Hansen, Hausman, and Newey}{Hansen
  et~al.}{2008}]{hansen2008estimation}
Hansen, C., J.~Hausman, and W.~Newey (2008).
\newblock Estimation with many instrumental variables.
\newblock {\em Journal of Business \& Economic Statistics\/}~{\em 26\/}(4),
  398--422.

\bibitem[\protect\citeauthoryear{Hansen and Kozbur}{Hansen and
  Kozbur}{2014}]{hansen2014instrumental}
Hansen, C. and D.~Kozbur (2014).
\newblock Instrumental variables estimation with many weak instruments using
  regularized {JIVE}.
\newblock {\em Journal of Econometrics\/}~{\em 182\/}(2), 290--308.

\bibitem[\protect\citeauthoryear{Hastie, Montanari, Rosset, and
  Tibshirani}{Hastie et~al.}{2022a}]{hastie2022supplement}
Hastie, T., A.~Montanari, S.~Rosset, and R.~J. Tibshirani (2022a).
\newblock Supplement to ``{S}urprises in high-dimensional ridgeless least
  squares interpolation''.
\newblock {\em Annals of Statistics\/}~{\em 50\/}(2), 949.

\bibitem[\protect\citeauthoryear{Hastie, Montanari, Rosset, and
  Tibshirani}{Hastie et~al.}{2022b}]{hastie2022surprises}
Hastie, T., A.~Montanari, S.~Rosset, and R.~J. Tibshirani (2022b).
\newblock Surprises in high-dimensional ridgeless least squares interpolation.
\newblock {\em Annals of Statistics\/}~{\em 50\/}(2), 949.

\bibitem[\protect\citeauthoryear{Koles{\'a}r, Chetty, Friedman, Glaeser, and
  Imbens}{Koles{\'a}r et~al.}{2015}]{kolesar2015identification}
Koles{\'a}r, M., R.~Chetty, J.~Friedman, E.~Glaeser, and G.~W. Imbens (2015).
\newblock Identification and inference with many invalid instruments.
\newblock {\em Journal of Business \& Economic Statistics\/}~{\em 33\/}(4),
  474--484.

\bibitem[\protect\citeauthoryear{Marchenko and Pastur}{Marchenko and
  Pastur}{1967}]{marchenko1967distribution}
Marchenko, V.~A. and L.~A. Pastur (1967).
\newblock Distribution of eigenvalues for some sets of random matrices.
\newblock {\em Matematicheskii Sbornik\/}~{\em 114\/}(4), 507--536.

\bibitem[\protect\citeauthoryear{Nagar}{Nagar}{1959}]{nagar1959bias}
Nagar, A.~L. (1959).
\newblock The bias and moment matrix of the general k-class estimators of the
  parameters in simultaneous equations.
\newblock {\em Econometrica\/}, 575--595.

\bibitem[\protect\citeauthoryear{Okui}{Okui}{2011}]{okui2011instrumental}
Okui, R. (2011).
\newblock Instrumental variable estimation in the presence of many moment
  conditions.
\newblock {\em Journal of Econometrics\/}~{\em 165\/}(1), 70--86.

\bibitem[\protect\citeauthoryear{Silverstein}{Silverstein}{1995}]{silverstein1995strong}
Silverstein, J.~W. (1995).
\newblock Strong convergence of the empirical distribution of eigenvalues of
  large dimensional random matrices.
\newblock {\em Journal of Multivariate Analysis\/}~{\em 55\/}(2), 331--339.

\bibitem[\protect\citeauthoryear{Spiess}{Spiess}{2017}]{spiess2017bias}
Spiess, J. (2017).
\newblock Bias reduction in instrumental variable estimation through
  first-stage shrinkage.
\newblock {\em arXiv preprint arXiv:1708.06443\/}.

\bibitem[\protect\citeauthoryear{Staiger and Stock}{Staiger and
  Stock}{1997}]{staiger1997instrumental}
Staiger, D. and J.~H. Stock (1997).
\newblock Instrumental variables regression with weak instruments.
\newblock {\em Econometrica\/}~{\em 65\/}(3), 557--586.

\bibitem[\protect\citeauthoryear{Stock and Wright}{Stock and
  Wright}{2000}]{stock2000gmm}
Stock, J.~H. and J.~H. Wright (2000).
\newblock {GMM} with weak identification.
\newblock {\em Econometrica\/}~{\em 68\/}(5), 1055--1096.

\end{thebibliography}



\begin{appendices}

\newpage
\section{Proof of Lemmas \label{appendix_proof_lemmas}}

\subsection*{Proof of Lemma \ref{Lemma_ridgerep}}

\begin{proof} \vspace{-0.5cm}
	\begin{eqnarray*}
		I_n&=& \ZZlamout\ZZlamout^{-1} \\
		&=& (ZZ'/n)\ZZlamout^{-1}+\lambda \ZZlamout^{-1} \\
		&=& Z\ZZlam^{-1}Z'/n + \lambda \ZZlamout^{-1}
	\end{eqnarray*}
	The last equation, which proves the Lemma, uses the so-called ``kernel trick'', i.e. $\ZZlam^{-1} Z' = Z'\ZZlamout^{-1}\,$.
\end{proof}

\subsection*{Proof of Lemma \ref{Lemma_companion}}

\begin{proof}
	This identity can be found in many textbooks \citep[see e.g. Lemma 2.4 in][]{couillet2022random}. It follows because the ESDs of $\widehat\Sigma$ and $\underline{\widehat\Sigma}$ have the same nonzero eigenvalues and differ only by $\vert n-k \vert$ zero eigenvalues for the larger matrix, that is
	\begin{eqnarray*}
		v\lam = \gamma m\lam + (1-\gamma)\frac{1}{\lambda} \, .
	\end{eqnarray*}
	Multiplying by $\lambda$ and rearranging leads to the claimed formula. 
\end{proof}

\subsection*{Proof of Lemma \ref{Lemma_vlam0}}

\begin{proof}
	The proof follows from Lemma 2.3 in \cite{dobriban_wager}. They use the dominated convergence theorem to swap limits and expectations, which allows us to take the limit as $\lambda\downarrow 0$ inside the expectation. If $\gamma<1$, we know that $Y$ is bounded away from 0 ($Y>c>0$ for some $c$) and the dominated convergence theorem is hence applicable to $m\lam$. If $\gamma>1$, we know that $\underline{Y}$ is bounded away from 0 ($\underline{Y}>c>0$ for some $c$) and the dominated convergence theorem is applicable to $v\lam$. We first consider the case with $\gamma<1$, then
	\begin{eqnarray*}
		\lim_{\lambda\downarrow 0}\lambda m\lam=\lim_{\lambda\downarrow 0} \E \left[ \frac{\lambda}{Y+\lambda} \right]=\E \left[ \frac{0}{Y} \right] =0 \, ,
	\end{eqnarray*}
	which proves 1b). 1a) follows immediately from this result using Lemma \ref{Lemma_companion}. Similarly, if $\gamma>1$, then 
	\begin{eqnarray*}
		\lim_{\lambda\downarrow 0}\lambda v\lam=\lim_{\lambda\downarrow 0} \E \left[ \frac{\lambda}{\underline{Y}+\lambda} \right]=\E \left[ \frac{0}{\underline{Y}} \right] =0 \, ,
	\end{eqnarray*}
	which proves 2a). 2b) follows immediately from this result using Lemma \ref{Lemma_companion}.
\end{proof}

\newpage
\section{Proof of Theorem \ref{theorem_ridge} \label{appendix_proof_ridge}}

\begin{proof}
	By definition of the linear model (\ref{model}), we can write
	\begin{align*}
		\widetilde\beta_{\lambda} = \beta + \frac{x'P_{\lambda}\varepsilon}{x'P_{\lambda}x } \, .
	\end{align*}
	To establish convergence of the remainder term, we argue that the quadratic terms involving $P_{\lambda}$ concentrate around their means and then use results from RMT to find the limits of those means. We begin with the denominator. 
	
	Following \cite{hastie2022supplement}, let $x = Z\pi + \nu$ and $c_n = \sqrt{k}(\sigma_\nu/\alpha)$ and observe that
	\begin{align*}
		\frac{x'P_{\lambda}x}{n} &= (\pi, \nu)^{\prime} \left(\frac{1}{n} \begin{bmatrix} Z \\ I_n \end{bmatrix}^{\prime} P_\lambda \begin{bmatrix} Z \\ I_n \end{bmatrix}\right)(\pi, \nu) \\
		&= \underbrace{(\pi, \nu/c_n)^{\prime}}_{\delta^{\prime}} \underbrace{\left(\frac{1}{n} \begin{bmatrix} Z \\ c_nI_n \end{bmatrix}^{\prime} P_\lambda \begin{bmatrix} Z \\ c_nI_n \end{bmatrix}\right)}_{A} \underbrace{(\pi, \nu/c_n)}_{\delta}.
	\end{align*}
	Because $\delta$ and $A$ are independent and $\delta$ has independent entries with mean zero and common variance $\alpha^2/k$, we can use the almost sure convergence of quadratic forms, from Lemma C.3 in \cite{dobriban_wager_supplement}, which is adapted from Lemma B.26 in \cite{bai2010spectral}. This result asserts that, almost surely,
	\begin{equation*}
		\delta^{\prime} A\delta - \left(\frac{\alpha^2}{k}\right)\tr{(A)} \to_{a.s.} 0 \, .
	\end{equation*}
	Now examine
	\begin{align*}
		\frac{\alpha^2}{k}\tr{(A)}  &= \frac{\alpha^2}{k} \tr\left(P_\lambda \left(\frac{Z Z'}{n} + \frac{c_n^2}{n}I_n\right)\right) \\
		&=\frac{\alpha^2}{k} \tr\left(Z'P_\lambda Z/n \right) + \frac{\sigma_\nu^2}{n} \tr\left(P_\lambda\right) \\
		&= \underbrace{\frac{\alpha^2}{k} \tr\left(Z'Z/n \right)}_{a} - \underbrace{\frac{\alpha^2\lambda}{k} \tr\big(Z' (ZZ'/n+\lambda I_n)^{-1}  Z/n \big)}_{b} + \underbrace{\frac{\sigma_\nu^2}{n} \tr\left(P_\lambda\right)}_{c} \, .
	\end{align*}
	where the first and second equalities follow from cyclic rotations within the trace operator and the third equality uses Lemma \ref{Lemma_ridgerep}. In the following, we consider $a,b$, and $c$ separately. For convenience, we use the Stieltjes transform $m\lam$ for all terms related to the identification strength $\alpha^2$ and the companion Stieltjes transform $v\lam$ for all terms related to $\sigma^2_\nu$ or $\sigma_{\varepsilon\nu}$. It is known that 
	\begin{equation*}
		a = \frac{\alpha^2}{k} \tr\left(Z'Z/n \right) \to_{a.s.} \alpha^2 \, .
	\end{equation*}
	
	Moreover,
	\begin{align*}
		b &= \frac{\alpha^2\lambda}{k} \tr\big(Z' (ZZ'/n+\lambda I_n)^{-1}   Z/n \big) \\
		&= \frac{\alpha^2\lambda}{k} \tr\big((Z'Z/n) (Z'Z/n+\lambda I_k)^{-1}  \big) \\
		&= \frac{\alpha^2\lambda}{k} \tr\big(I_k - \lambda (Z'Z/n+\lambda I_k)^{-1}  \big)\to_{a.s.} \alpha^2\lambda \left(1-\lambda m(-\lambda) \right) \, ,
	\end{align*}
	where the second equality uses the kernel trick, and the third equality follows from the relation $Z'Z/n=(Z'Z/n+\lambda I_k)-\lambda I_k\, $. 
	
	Finally, by similar calculations, we have
	\begin{align*}
		c &= \frac{\sigma_\nu^2}{n} \tr\big( Z(Z'Z/n+\lambda I_k)^{-1}Z'/n \big) \\
		&= \frac{\sigma_\nu^2}{n} \tr\big( I_n - \lambda (ZZ'/n+\lambda I_n)^{-1} \big) \\
		& \to_{a.s.} \sigma_\nu^2  (1-\lambda v\lam) \, .
	\end{align*}
	Hence, we have shown that, almost surely,
	\begin{equation*}
		\frac{x'P_\lambda x}{n} \to_{a.s.} \alpha^2 \big(1-\lambda(1-\lambda m\lam)\big) + \sigma_\nu^2 (1-\lambda v\lam) \, .
	\end{equation*}
	
	For the numerator, we get
	\begin{align*}
		\frac{x' P_{\lambda} \varepsilon}{n} &= (Z\pi+\nu)'P_\lambda \varepsilon/n \\
		&= \pi'Z'P_\lambda \varepsilon/n + \frac{1}{n}\tr\big(\nu'( I_n - \lambda (ZZ'/n+\lambda I_n)^{-1} )\varepsilon\big) \\
		&\to_{a.s.} \sigma_{\varepsilon\nu} \big( 1 - \lambda v \lam \big) \, .
	\end{align*}
	where the second equality uses Lemma \ref{Lemma_ridgerep}. Hence, by the continuous mapping theorem, we have
	\begin{align*}
		\widetilde\beta_{\lambda} - \beta \to_{a.s.} \frac{ \sigma_{\varepsilon\nu} \big( 1 - \lambda v \lam \big) }{ \alpha^2 \big(1-\lambda(1-\lambda m\lam)\big) + \sigma_\nu^2 (1-\lambda v\lam) } \, .
	\end{align*}
\end{proof}

\newpage
\section{Proof of Theorem \ref{theorem_ba_ridge} \label{appendix_proof_ba_ridge}}

\begin{proof}	
	By definition of the linear model (\ref{model}), we can write
	\begin{align*}
		\widehat\beta_{\lambda} = \beta + \frac{x'S_{\lambda}\varepsilon}{x'S_{\lambda}x } \, .
	\end{align*}
	Recall the definition of $S_{\lambda}=P_{\lambda} - (1-\lambda \widehat{v}\lam) I_{n} \, $. From Theorem \ref{theorem_ridge} we know for the first term of the denominator that
	\begin{equation*}
		\frac{x'P_\lambda x}{n} \to_{a.s.} \alpha^2 \big(1-\lambda p\lam\big) + \sigma_\nu^2 (1-\lambda v\lam) \, .
	\end{equation*}
	where we use the definition $p\lam=1-\lambda m\lam$ for clarity. 
	
	Using similar calculations as in Theorem \ref{theorem_ridge}, we get the second term
	\begin{align*}
		\frac{x'x}{n} &= (\pi, \nu)^{\prime} \left(\frac{1}{n} \begin{bmatrix} Z \\ I_n \end{bmatrix}^{\prime} \begin{bmatrix} Z \\ I_n \end{bmatrix}\right)(\pi, \nu) \\
		&= \underbrace{(\pi, \nu/c_n)^{\prime}}_{\delta^{\prime}} \underbrace{\left(\frac{1}{n} \begin{bmatrix} Z \\ c_n I_n \end{bmatrix}^{\prime} \begin{bmatrix} Z \\ c_n I_n \end{bmatrix}\right)}_{A} \underbrace{(\pi, \nu/c_n)}_{\delta}.
	\end{align*}
	Again, because $\delta$ and $A$ are independent and $\delta$ has independent entries with mean zero and common variance $\alpha^2/k$, we can use the almost sure convergence of quadratic forms, from Lemma C.3 in \cite{dobriban_wager_supplement}, which is adapted from Lemma B.26 in \cite{bai2010spectral}.
	Therefore, we have
	\begin{align*}
		\frac{\alpha^2}{k}\tr{(A)}  &= \frac{\alpha^2}{k} \tr\left(I_n \left(\frac{Z Z'}{n} + \frac{c_n^2}{n}I_n\right)\right) \\
		&= \underbrace{\frac{\alpha^2}{k} \tr\big(Z'Z/n)}_{a} + \underbrace{\frac{\sigma_\nu^2}{n} \tr\left(I_n \right)}_{b} \, ,
	\end{align*}
	where the first and second equality follow from cyclic rotations within the trace operator. 
	We have
	\begin{align*}
		a &= \frac{\alpha^2}{k} \tr\big(Z'Z/n \big) \to_{a.s.} \alpha^2 \, ,
	\end{align*}
	and 
	\begin{align*}
		b &= \frac{\sigma_\nu^2}{n} \tr\left(I_n \right) \to_{a.s.} \sigma_\nu^2 \, .
	\end{align*}	
	Therefore, the denominator of $\widehat\beta$ converges almost surely to
	\begin{align*}
		\frac{x'S_\lambda x}{n}  &\to_{a.s.}  \alpha^2 \big( 1-\lambda p\lam \big) + \sigma_\nu^2 \big( 1-\lambda v\lam \big) - (1-\lambda v\lam \big)[ \alpha^2 + \sigma_\nu^2] \, \\
		&= \alpha^2 \big( 1-\gamma p\lam - \lambda p\lam \big).
	\end{align*}
For the numerator we have 
\begin{align*}
		\frac{x' S_{\lambda} \varepsilon}{n} &\to_{a.s.} \sigma_{\varepsilon\nu} \big( 1 - \lambda v \lam \big) - \big( 1 - \lambda v \lam \big)\sigma_{\varepsilon\nu} =0 \, .
	\end{align*}
	Hence, by the continuous mapping theorem, we finally have
	\begin{align*}
		\widehat\beta_{\lambda} - \beta \to_{a.s.} 0 \, .
	\end{align*}
\end{proof}

\newpage
\section{Proof of Corollaries \label{appendix_proof_corollaries}}

To simplify the illustration, we use the following notation in the proofs of the Corollaries: $m'\lam=\E\left[ \frac{1}{(Y+\lambda)^2} \right]$ and $-p'\lam=m\lam-\lambda m'\lam \, $. 

\subsection*{Proof of Corollary \ref{corollary_amplifier}}

\begin{proof}	
The first claim, $\lim_{\lambda\downarrow 0} a\lam =1$, follows immediately from Lemma \ref{Lemma_vlam0}. For the second claim, note that
\begin{align*}
	a\lam = \gamma\frac{1-\lambda p\lam}{1-\lambda v\lam}=\frac{1-\lambda p\lam}{p\lam}=\frac{1}{p\lam}-\lambda 
\end{align*}
which follows from Lemma \ref{Lemma_companion}. For the first derivative, we see
\begin{align*}
	a'\lam=-\frac{p'\lam}{p^2\lam}-1 \, ,
\end{align*}
which, for $v\lam\neq 1$ or equivalently $m\lam\neq 1$,  is positive if
\begin{align*}
	m\lam - \lambda m'\lam - 1 + 2\lambda m\lam - \lambda^2m^2\lam \geq 0 \, ,
\end{align*}
or, equivalently,
\begin{align*}
	\underbrace{m\lam - \lambda m'\lam - 1 + 2\lambda m\lam - \lambda^2 m'\lam}_{A} + \underbrace{\lambda^2 m'\lam - \lambda^2m^2\lam}_{B} \geq 0 \, .
\end{align*}
Note that 
\begin{align*}
	B &= \E\left[ \left(\frac{\lambda}{Y+\lambda}\right)^2 \right] - \E\left[ \frac{\lambda}{Y+\lambda} \right]^2 \geq 0
\end{align*} 
is a variance term and hence positive, and 
\begin{align*}
	A &= \E\left[ \frac{1}{Y+\lambda} \right]-\E\left[ \frac{\lambda+\lambda^2}{(Y+\lambda)^2}\right]+\E\left[ \frac{2\lambda}{Y+\lambda} \right] - 1  \\
	&= \E\left[ \frac{(1+2\lambda)(Y+\lambda)-\lambda-\lambda^2-(Y^2+2\lambda Y + \lambda^2)}{(Y+\lambda)^2}\right] \\
	&= \E\left[ \frac{Y-Y^2}{(Y+\lambda)^2}\right] .
\end{align*}
Moreover, using the dominated convergence theorem, we get
\begin{align*}
	\lim_{\lambda\downarrow 0} \E\left[ \frac{Y-Y^2}{(Y+\lambda)^2}\right] =E[Y^{-1}]-1 \geq 0 \, ,
\end{align*}
which proves the claim.
\end{proof}

\subsection*{Proof of Corollary \ref{corollary_badj_derivative}}

\begin{proof}	
The first claim, $\lim_{\lambda\downarrow 0} f(\lambda)=1-\gamma$, follows immediately from Lemma \ref{Lemma_vlam0}. For the second claim, note that
\begin{align*}
	f'\lam &= -\big(\lambda+\gamma\big)p'\lam - p\lam \, 
\end{align*}
Therefore, we get
\begin{align*}
	f'\lam = \E\left[ \frac{(\lambda+\gamma)Y}{(Y+\lambda)^2} \right]-\E\left[ \frac{Y}{Y+\lambda} \right] =\gamma \E\left[ \frac{Y}{(Y+\lambda)^2} \right] - \E\left[ \frac{Y^2}{(Y+\lambda)^2} \right] \, ,
\end{align*}
Moreover, using the dominated convergence theorem, we get
\begin{align*}
	\lim_{\lambda\downarrow 0} \left( \gamma \E\left[ \frac{Y}{(Y+\lambda)^2} \right] - \E\left[ \frac{Y^2}{(Y+\lambda)^2} \right]\right) = \gamma \E\left[ Y^{-1} \right] - 1\geq 0 \, ,
\end{align*}
which proves the claim. 
\end{proof}

\newpage
\section{Asymptotic Variance of $\widehat\beta_\lambda$ \label{appendix_variance}}

Recall the definition of $S_{\lambda}=P_{\lambda} - (1-\lambda \widehat{v}\lam) I_{n} \, $ and, therefore, $S_\lambda S_\lambda = P_\lambda P_\lambda - 2(1-\lambda \widehat{v}\lam)P_\lambda + (1-\lambda\widehat{v}\lam)^2 I_n\, $. Following \cite{hastie2022supplement}, let $x = Z\pi + \nu$ and $c_n = \sqrt{k}(\sigma_\nu/\alpha)$ and observe that
	\begin{align*}
		\frac{x'P_{\lambda}P_{\lambda}x}{n} &= (\pi, \nu)^{\prime} \left(\frac{1}{n} \begin{bmatrix} Z \\ I_n \end{bmatrix}^{\prime} P_\lambda P_\lambda \begin{bmatrix} Z \\ I_n \end{bmatrix}\right)(\pi, \nu) \\
		&= \underbrace{(\pi, \nu/c_n)^{\prime}}_{\delta^{\prime}} \underbrace{\left(\frac{1}{n} \begin{bmatrix} Z \\ c_nI_n \end{bmatrix}^{\prime} P_\lambda P_\lambda \begin{bmatrix} Z \\ c_nI_n \end{bmatrix}\right)}_{A} \underbrace{(\pi, \nu/c_n)}_{\delta}.
	\end{align*}
	Because $\delta$ and $A$ are independent and $\delta$ has independent entries with mean zero and common variance $\alpha^2/k$, we can use the almost sure convergence of quadratic forms, from Lemma C.3 in \cite{dobriban_wager_supplement}, which is adapted from Lemma B.26 in \cite{bai2010spectral}. This result asserts that, almost surely,
	\begin{equation*}
		\delta^{\prime} A\delta - \left(\frac{\alpha^2}{k}\right)\tr{(A)} \to_{a.s.} 0 \, .
	\end{equation*}
	Now examine
	\begin{align*}
		\frac{\alpha^2}{k}\tr{(A)}  &= \frac{\alpha^2}{k} \tr\left(P_\lambda P_\lambda \left(\frac{Z Z'}{n} + \frac{c_n^2}{n}I_n\right)\right) \\
		&=\frac{\alpha^2}{k} \tr\left(Z'P_\lambda P_\lambda Z/n \right) + \frac{\sigma_\nu^2}{n} \tr\left(P_\lambda P_\lambda\right) \\
		&=\frac{\alpha^2}{k} \tr\left(Z'(I_n -2 \lambda(ZZ'/n + \lambda I_n)^{-1} +\lambda^2(ZZ'/n + \lambda I_n)^{-2})Z/n \right) + \frac{\sigma_\nu^2}{n} \tr\left(P_\lambda P_\lambda\right) \\
		&= \underbrace{\frac{\alpha^2}{k} \tr\left(Z'Z/n \right)}_{a} - \underbrace{2 \frac{\alpha^2\lambda}{k} \tr\big(Z' (ZZ'/n+\lambda I_n)^{-1}  Z/n \big)}_{b}  \\ &+  \underbrace{\frac{\alpha^2\lambda^2}{k} \tr\big(Z' (ZZ'/n+\lambda I_n)^{-2}  Z/n \big)}_{c} + \underbrace{\frac{\sigma_\nu^2}{n} \tr\left(P_\lambda P_\lambda\right)}_{d} \, .
	\end{align*}
	where the first and second equality follow from cyclic rotations within the trace operator and the third equality uses Lemma \ref{Lemma_ridgerep}. In the following, we consider $a,b,c$, and $d$ separately. For convenience, we use the Stieltjes transform $m\lam$ and $p\lam=1-\lambda m\lam$ for all terms related to the identification strength $\alpha^2$ and the companion Stieltjes transform $v\lam$ for all terms related to $\sigma^2_\nu$ or $\sigma_{\varepsilon\nu}$. It is known that 
	\begin{equation*}
		a = \frac{\alpha^2}{k} \tr\left(Z'Z/n \right) \to_{a.s.} \alpha^2 \, .
	\end{equation*}
	
	Moreover,
	\begin{align*}
		b &= 2\frac{\alpha^2\lambda}{k} \tr\big(Z' (ZZ'/n+\lambda I_n)^{-1}   Z/n \big) \\
		&= 2\frac{\alpha^2\lambda}{k} \tr\big((Z'Z/n) (Z'Z/n+\lambda I_k)^{-1}  \big) \\
		&= 2\frac{\alpha^2\lambda}{k} \tr\big(I_k - \lambda (Z'Z/n+\lambda I_k)^{-1}  \big)\to_{a.s.} 2\alpha^2\lambda p\lam \, ,
	\end{align*}
	where the second equality uses the kernel trick, and the third equality follows from the relation $Z'Z/n=(Z'Z/n+\lambda I_k)-\lambda I_k\, $. Likewise, 
	\begin{align*}
		c &= \frac{\alpha^2\lambda^2}{k} \tr\big(Z' (ZZ'/n+\lambda I_n)^{-2}   Z/n \big) \\
		&= \frac{\alpha^2\lambda^2}{k} \tr\big((Z'Z/n) (Z'Z/n+\lambda I_k)^{-2}  \big) \\
		&= \frac{\alpha^2\lambda^2}{k} \tr\big((Z'Z/n + \lambda I_k)^{-1} - \lambda (Z'Z/n+\lambda I_k)^{-2} \big)\\ &\to_{a.s.} - \alpha^2\lambda^2 p'\lam \, ,
	\end{align*}
	where $p'\lam=\lambda m'(-\lambda)-m(-\lambda) \, $.
	Finally, by similar calculations, we have
	\begin{align*}
		d &= \frac{\sigma_\nu^2}{n} \tr\big( (Z(Z'Z/n+\lambda I_k)^{-1}Z'/n)(Z(Z'Z/n+\lambda I_k)^{-1}Z'/n) \big) \\
				 &= \frac{\sigma_\nu^2}{n} \tr\big(I_n -2 \lambda(ZZ'/n + \lambda I_n)^{-1}+ \lambda^2 (ZZ'/n + \lambda I_n)^{-2}  \big) \\
		& \to_{a.s.} \sigma_\nu^2  (1-2 \lambda v\lam + \lambda^2 v'\lam) \, .
	\end{align*}
	Hence, we have shown that, almost surely,
	\begin{equation*}
		\frac{x'P_\lambda P_\lambda x}{n} \to_{a.s.} \alpha^2 \big(1-2 \lambda p\lam - \lambda^2p'\lam \big) + \sigma_\nu^2 (1-2 \lambda v\lam + \lambda^2 v'\lam) \, .
	\end{equation*}
We suppress the $\lam$ in the following to simplify the notation. Moreover,
\begin{align*}
	\frac{x'S_\lambda S_\lambda x}{n} &= \frac{x'P_\lambda P_\lambda x}{n} -2 (1-\lambda \widehat{v}) \, \frac{x'P_\lambda x}{n} + (1-\lambda \widehat{v})^2 \, \frac{x' x}{n} \\
	&\to_{a.s.} \alpha^2 \big(1-2 \lambda p - \lambda^2 p' \big) + \sigma_\nu^2 (1-2 \lambda v + \lambda^2 v') \\ & - 2(1-\lambda v)  \big[\alpha^2 (1- \lambda p ) + \sigma_\nu^2 (1- \lambda v)  \big] + (1- \lambda v)^2\big[\alpha^2 + \sigma_\nu^2 \big] \\
	&=  \alpha^2 \big[ 1 - 2 \lambda p - \lambda^2 p' - 2 \gamma p(1-\lambda p) +\gamma^2 p^2 \big] + \sigma_\nu^2 \big[ 1-2\lambda v+\lambda^2 v' - (1-\lambda v)^2 \big] \\	
	&= \alpha^2 \big[ 1-2(\lambda+\gamma)p + 2\gamma\lambda p^2 + \gamma^2 p^2 + \lambda^2 p^2 - \lambda^2 p^2 - \lambda^2 p' \big] + \sigma_\nu^2 \big[ \lambda^2 v' - \lambda^2 v^2 \big] \\	
	&=  \alpha^2 \big[ f^2 - \lambda^2 (p^2 + p') \big] + \sigma_\nu^2 \big[ \lambda^2 v' - \lambda^2 v^2 \big] \, ,
\end{align*}
where the first equality uses $1-\lambda v = \gamma p$ and the third equality uses $f^2=(1-(\lambda+\gamma)p)^2 \, $.
First, consider the numerator of $\widehat\V(\widehat\beta_{\lambda})$, 
\begin{align*}
	\frac{x'S_\lambda S_\lambda\widetilde{x}}{n} &= \frac{x'P_\lambda P_\lambda \widetilde{x}}{n} - 2(1-\lambda \widehat{v}) \, \frac{x'P_\lambda \widetilde{x}}{n} + (1-\lambda \widehat{v})^2 \, \frac{x' \widetilde{x}}{n} \\
	&= \frac{x'S_\lambda S_\lambda x}{n} + \frac{1}{n}\big( x'P_\lambda P_\lambda \widehat\varepsilon - 2(1-\lambda\widehat{v}) x'P_\lambda \widehat\varepsilon + (1-\lambda\widehat{v})^2 x'\widehat\varepsilon \, \big)\frac{\widehat\varepsilon'x}{\widehat\varepsilon'\widehat\varepsilon}  \\
	&= \frac{x'S_\lambda S_\lambda x}{n} + \frac{1}{n}\big( x'P_\lambda P_\lambda \widehat\varepsilon - (1-\lambda\widehat{v})^2 x'\widehat\varepsilon \, \big)\frac{\widehat\varepsilon'x}{\widehat\varepsilon'\widehat\varepsilon}  \\
	&= \frac{x'S_\lambda S_\lambda x}{n} + \big(\lambda^2 \widehat{v}' - \lambda^2\widehat{v}^2 \big) \frac{x'\widehat\varepsilon}{n} \, \frac{\widehat\varepsilon'x}{\widehat\varepsilon'\widehat\varepsilon}  \\
	&\to_{a.s.} \alpha^2 \big[ f^2 - \lambda^2 (p^2+p') \big] + \left( \sigma_\nu^2 + \frac{\sigma_{\varepsilon\nu}^2}{\sigma^2_\varepsilon} \right) \big[ \lambda^2 v' - \lambda^2 v^2 \big] \, ,
\end{align*}
where we use $S_\lambda S_\lambda = P_\lambda P_\lambda - 2(1-\lambda \widehat{v})P_\lambda + (1-\lambda\widehat{v})^2 I_n$ in the second equality, $x'P_\lambda\widehat\varepsilon = (1-\lambda\widehat{v}) x'\widehat\varepsilon$ in the third equality, and $x'P_\lambda P_\lambda \widehat\varepsilon = (1-2\lambda\widehat{v}+\lambda^2 \widehat{v}') x'\widehat\varepsilon$ in the fourth equality. From Theorem \ref{theorem_ba_ridge}, we have $x'S_\lambda x/n \to_{a.s.} \alpha^2 f \, $. Therefore, we finally get 
\begin{align*}
	\widehat\V(\widehat\beta_\lambda) \to_{a.s.} \frac{1}{n} \, \left( \frac{\sigma_\varepsilon^2}{\alpha^2} \bigg[ 1 - \frac{ \lambda^2 (p^2 + p')}{f^2} \bigg] + \frac{\left( \sigma^2_\varepsilon \sigma_\nu^2 + \sigma_{\varepsilon\nu}^2 \right)}{\alpha^4}\frac{ \big[ \lambda^2 v' - \lambda^2 v^2 \big]}{f^2} \right) \, .
\end{align*}

We can further rewrite the equation by factoring out $ \sigma^2_\varepsilon \sigma_\nu^2 / \alpha^4$ to obtain Equation \ref{eq:var} from the main part of the paper:

\begin{align*}
	\widehat\V(\widehat\beta_\lambda) \to_{a.s.} \frac{1}{n}\frac{\sigma^2_\varepsilon \sigma_\nu^2}{\alpha^4} \, \left( \gamma F - \frac{ \gamma F \lambda^2 \big[ p' + p^2 \big]}{f^2} + \frac{ \lambda^2 \big[ v' - v^2 \big]}{f^2} \left( 1 + \rho_{\varepsilon\nu}^2 \right) \right) \, .
\end{align*}

For $\lambda=0$ and $\gamma<1$, we get the Bekker asymptotic variance as a special case because then $\lim_{\lambda\downarrow 0} p\lam = 1$, $\lim_{\lambda\downarrow 0} p'\lam = \E[Y^{-1}]$, $\lim_{\lambda\downarrow 0} f\lam = 1-\gamma$ and 
\begin{align*}
	\lim_{\lambda\downarrow 0} \lambda^2 v'\lam - \lambda^2 v^2\lam &= \lim_{\lambda\downarrow 0} \lambda^2 v^2\lam\big(\widehat{v}'\lam/\widehat{v}^2\lam-1\big) \\
	&=(1-\gamma)^2\left( \frac{1}{1-\gamma}-1 \right)=\gamma(1-\gamma) \, .
\end{align*}
Therefore,
\begin{align*}
	\V(\widehat\beta_0)  \to_{a.s.} \frac{1}{n}\frac{\sigma^2_\varepsilon \sigma_\nu^2}{\alpha^4} \, \left( \gamma F + \frac{\gamma}{1-\gamma} \left( 1 + \rho_{\varepsilon\nu}^2 \right) \right) \, .
\end{align*}


\newpage

\section{Asymptotic Variance Comparison \label{appendix_proof_variance_standard_ba2SLS}}

Below, we prove that the asymptotic variance of the bias-adjusted 2SLS-Ridge estimator is smaller than the asymptotic variance of the standard bias-adjusted 2SLS estimator for some values of $\lambda$. The proof proceeds in two steps: First, we compute the derivative of the asymptotic variance of the bias-adjusted 2SLS-Ridge estimator with respect to $\lambda$. Second, we show that this derivative is negative in the  $\lambda \downarrow 0$ limit.
\bigskip 

\begin{proof}
\begin{align*}
	&\frac{\partial}{\partial \lambda} \left\{ \frac{1}{n} \, \left( \frac{\sigma_\varepsilon^2}{\alpha^2} \bigg[ 1 - \frac{ \lambda^2 (p^2 + p')}{f^2} \bigg] + \frac{\left( \sigma^2_\varepsilon \sigma_\nu^2 + \sigma_{\varepsilon\nu}^2 \right)}{\alpha^4}\frac{ \big[ \lambda^2 v' - \lambda^2 v^2 \big]}{f^2} \right) \right\}
\end{align*}

Considering the summands separately, we get for the first summand (the factor $\frac{\sigma_\varepsilon^2}{n \alpha^2}$ is omitted for ease of notation):

\begin{align*}
	&\frac{\partial}{\partial \lambda} \bigg[ 1 - \frac{ \lambda^2 (p^2 + p')}{f^2} \bigg] \\
	&= \frac{\lambda^2(p^2 + p') 2ff' - \left[2\lambda(p^2+p') + \lambda^2(2pp'+p'')\right] f^2}{f^4} \\
	&= \frac{- 2 \lambda^2(p^2 + p')  (\lambda + \gamma) p'  - 2 \lambda^2 (p^2+p') p -  2\lambda (p^2+p') }{f^3} \\
	&+ \frac{- \lambda^2 (2pp'+p'')  +  2 \lambda (\lambda + \gamma) p (p^2+p') + \lambda^2 (\lambda + \gamma) p (2pp'+p'')}{f^3} \\
	&= \frac{- 2 \lambda^2  (\lambda + \gamma) p^2 p' - 2 \lambda^2  (\lambda + \gamma) {p'}^2  -  2 \lambda^2 p^3 - 2 \lambda^2 pp' - 2 \lambda p^2  - 2 \lambda p'  }{f^3} \\
	&+ \frac{-2 \lambda^2 pp' - \lambda^2 p'' + 2 \lambda (\lambda+\gamma) p^3 + 2 \lambda (\lambda + \gamma) pp' + 2 \lambda^2 (\lambda + \gamma) p^2 p' + \lambda^2 (\lambda + \gamma) pp''}{f^3} \\
&= \frac{-2\lambda p' + (-2\lambda^2 + 2 \lambda \gamma) pp'- \lambda^2 p''+ \lambda^2(\lambda + \gamma) pp'' - 2\lambda^2 (\lambda + \gamma) {p'}^2- 2 \lambda p^2 + 2 \gamma \lambda p^3 }{f^3}		
\end{align*}

where we used $f = 1 - \lambda p - \gamma p = \lambda v - \lambda p =  \lambda (v-p)$ and $f' = - (\lambda + \gamma) p' -p $. \\

For the second summand, we get (the factor $\frac{\sigma_\varepsilon^2 \sigma_\nu^2 + \sigma_{\varepsilon\nu}^2}{n \alpha^4}$ is omitted for ease of notation):

\begin{align*}
	&\frac{\partial}{\partial \lambda} \left( \frac{ \lambda^2 v' - \lambda^2 v^2}{f^2} \right) \\
	&=  \frac{f^2 (2 \lambda v' - \lambda^2 v'' - 2 \lambda v^2  + \lambda^2 2 v v') - (\lambda^2 v' - \lambda^2 v^2 ) 2ff'}{f^4} \\
	&=  \frac{2 \lambda^2 v' (v-p) - \lambda^3 v'' (v-p) - 2 \lambda^2 v^2 (v-p) + \lambda^3 2 vv' (v-p)}{f^3} \\
		&- \frac{2 \lambda^2 v' (v-p) -  2 \lambda^2 v^2 (v-p) - 2 \lambda^3 v' (v' + p') + 2 \lambda^3 v^2 (v' + p')}{f^3} \\
				&= \frac{- \lambda^3 v'' v + \lambda^3 v'' p + 2 \lambda^3 v^2 v' - 2 \lambda^3 v v' p + 2 \lambda^3 {v'}^2 + 2 \lambda^3v'p' - 2\lambda^3 v^2 v' - 2 \lambda^3 v^2 p'}{f^3} \\
				&= \frac{-\lambda^3 v'' v + \lambda^3 v'' p  - 2 \lambda^3 v v' p + 2 \lambda^3 {v'}^2 + 2 \lambda^3v'p'  - 2 \lambda^3 v^2 p'}{f^3}  \\
\end{align*}
where we used $f = 1 - \lambda p - \gamma p = \lambda v - \lambda p =  \lambda (v-p)$ and $f' = (v-p) - \lambda (v' + p')$. \\

Now, we can investigate the behavior of both summands in the $\lambda \downarrow 0$ limit. The first summand vanishes, as there is at least one $\lambda$ in every summand and no diverging terms (see Lemma \ref{Lemma_vlam0}).

For the second summand, it is more complicated, as $\lambda v$ does not vanish. Some prior calculations are needed: 
 \begin{align*}
	& v = \frac{1}{\lambda} - \frac{\gamma}{\lambda} p \\
	& v' = \frac{1}{\lambda^2} - \frac{\gamma}{\lambda^2} p + \frac{\gamma}{\lambda} p' = \frac{1}{\lambda} v  + \frac{\gamma}{\lambda} p' \\
		& v'' = 2\frac{1}{\lambda^2} v  + 2 \frac{\gamma}{\lambda^2} p' - \frac{\gamma}{\lambda} p''  
\end{align*}

We know from Lemma \ref{Lemma_vlam0} that in the $\lambda \downarrow 0$ (and $\gamma < 1$) limit, the following holds:
  \begin{align*}
   & \lambda v = 1-\gamma, \  \lambda m = 0, \  p = 1, \ f = 1-\gamma \, .
\end{align*}

We now take, using the above calculations, the limit for all terms of the second summand separately:

\begin{align*}
   & \lim_{\lambda \downarrow 0, \gamma < 1} - \lambda^3 v'' v \\
   &= \lim_{\lambda \downarrow 0, \gamma < 1} - 2 \lambda v^2 - 2\gamma \lambda p' v + \gamma \lambda^2 p'' v \\
   &= \lim_{\lambda \to 0, \gamma < 1} - 2(\lambda v) v - 2 \gamma p' (\lambda v) + \gamma p'' (\lambda v) \lambda \\
   &= - 2 (1-\gamma) v + 2 \gamma (1-\gamma) m  \\
   \\
   & \lim_{\lambda \downarrow 0, \gamma < 1} \lambda^3 v'' p \\
   &= \lim_{\lambda \downarrow 0, \gamma < 1}  2 \lambda v p + 2 \gamma \lambda p p' - \gamma \lambda^2 p p'' \\
   &= 2 (1-\gamma) \\
   \\
   & \lim_{\lambda \downarrow 0, \gamma < 1} - 2 \lambda^3 v v' p \\
   &= \lim_{\lambda \downarrow 0, \gamma < 1} - 2 \lambda^2 v^2 p - 2 \lambda^2 \gamma p p' v \\
   &= - \lim_{\lambda \downarrow 0, \gamma < 1}  2 (\lambda v)^2 p - 2 \lambda (\lambda v) \gamma p p' \\
   &= - 2 (1-\gamma)^2 \\
   \\
   & \lim_{\lambda \downarrow 0, \gamma < 1}  2 \lambda^3 {v'}^2 \\
   &= \lim_{\lambda \downarrow 0, \gamma < 1}  2 \lambda^3 \left( - \frac{1}{\lambda} v - \frac{\gamma}{\lambda} p' \right)^2 \\
   &= \lim_{\lambda \downarrow 0, \gamma < 1}  2 \lambda^3 \left( \frac{1}{\lambda^2} v^2 + 2 \frac{\gamma}{\lambda^2} v p' + \frac{\gamma^2}{\lambda^2} {p'}^2 \right) \\
   &= \lim_{\lambda \downarrow 0, \gamma < 1}  2 ( \lambda v) v + 4 (\lambda v) \gamma p' + 2 \lambda \gamma^2 {p'}^2 \\
   &= 2 (1-\gamma) v - 4 (1-\gamma) \gamma m \\
\end{align*}

\begin{align*}
   & \lim_{\lambda \downarrow 0, \gamma < 1}  2 \lambda^3 v' p' \\
   &= \lim_{\lambda \downarrow 0, \gamma < 1} - 2 \lambda^2 v p' - 2 \lambda^2 \gamma {p'}^2 \\
   &= \lim_{\lambda \downarrow 0, \gamma < 1} - 2 \lambda (\lambda v) p' - 2 \lambda^2 \gamma {p'}^2 \\
   &= 0 \\
   \\
   & \lim_{\lambda \downarrow 0, \gamma < 1}  2 \lambda^3 v^2 p' \\
   &= \lim_{\lambda \downarrow 0, \gamma < 1}  2 \lambda (\lambda v)^2 p' \\
   &= 0
\end{align*}

In total we get: 

  \begin{align*}
   & \lim_{\lambda \downarrow 0, \gamma<1} \frac{-\lambda^3 v'' v + \lambda^3 v'' p  - 2 \lambda^3 v v' p + 2 \lambda^3 {v'}^2 + 2 \lambda^3v'p'  - 2 \lambda^3 v^2 p'}{f^3}    \\
   &=  (1-\gamma)^{-3} \left[ - 2 (1-\gamma) v + 2 \gamma (1-\gamma) m  + 2 (1-\gamma) - 2(1-\gamma)^2 +  2 (1-\gamma) v - 4  (1-\gamma) \gamma m \right] \\
   &= (1-\gamma)^{-3} \left[  - 2 \gamma (1-\gamma) m  - 2 \gamma (1-\gamma)   \right] 
\end{align*}

If this  term is negative, the bias-adjusted 2SLS-Ridge estimator has less variance than the standard bias-adjusted 2SLS estimator for some small, positive $\lambda$. We can show that this is indeed the case under the assumption of $E\left[Y^{-1}\right]  > 1 $:

  \begin{align*}
 &  - 2 \gamma (1-\gamma) m  - 2 \gamma (1-\gamma)  \\
   &= - 2 \gamma (1-\gamma) E\left[Y^{-1}\right]   - 2 \gamma (1-\gamma) \overset{!}{<} 0 \\
   &\iff   E\left[Y^{-1}\right] 
    \overset{!}{>} 1
\end{align*}

\begin{align*}
	&  \lim_{\lambda \downarrow 0, \gamma<1}  \big( - 2 \gamma (1-\gamma) m  - 2 \gamma (1-\gamma) \big) = \lim_{\lambda \downarrow 0, \gamma<1} \bigg( - 2 \gamma (1-\gamma) E\left[ \frac{1}{Y+\lambda} \right]   - 2 \gamma (1-\gamma) \bigg) \\
	& = - 2 \gamma (1-\gamma) E\left[Y^{-1}\right]   - 2 \gamma (1-\gamma) \overset{!}{<} 0 \\
	&\iff   E\left[Y^{-1}\right] 
	\overset{!}{>} 1 \, .
\end{align*}

\end{proof}

\clearpage
\section{Optimal Tuning Parameter}\label{appendix_proof_lambda_star}

We want to select a tuning-parameter $\lambda$ that minimizes the asymptotic variance of the bias-adjusted 2SLS-Ridge estimator. Therefore, our strategy is to compute the derivative of Equation \ref{eq:var}, set it to zero, and derive an expression for $\lambda$ corresponding to the optimal tuning-parameter.

The derivative of Equation \ref{eq:var} with respect to $\lambda$ is given by:

\begin{equation*}
    \begin{split}
          \frac{\partial \mathrm{Var}(\hat{\beta}_{\lambda})}{\partial \lambda} &= \frac{\sigma^2_\nu \sigma^2_{\epsilon} F }{\alpha^4 n} \frac{1}{f^3} \\
          &  \bigg\{ - \gamma \bigg( \left[ 2 \lambda (p'+p^2) + \lambda^2 (p'' + 2 p p') \right] f - 2 \lambda^2 (p'+p^2) f' \bigg) \\
          & \ \ \ + \bigg( \left[ 2 \lambda (v'-v^2) + \lambda^2(-v''+2vv')\right] f - 2 \lambda^2 (v'-v^2) f' \bigg) * \frac{1+\rho_{\epsilon \nu}^2}{F} \bigg\}
    \end{split}
\end{equation*}

We rewrite this equation in terms of v, v', v'', using:

\begin{equation}\label{formulas_p}
    \begin{split}
          p &= \frac{1}{\gamma} (1-\lambda v) \\
          p'&= -\frac{1}{\gamma} (v - \lambda v') \\
          p'' &= \frac{1}{\gamma} (2 v' - \lambda v'') \\
          f &= \lambda \left(1 + \frac{\lambda}{\gamma} \right)  v - \frac{\lambda}{\gamma} \\
          f' &= \left(1 + 2\frac{\lambda}{\gamma} \right) v - \left(1 + \frac{\lambda}{\gamma} \right) v' - \frac{1}{\gamma}
    \end{split}
\end{equation}

\clearpage

Plugging in Equations \ref{formulas_p} yields:

\begin{equation*}
    \begin{split}
          \frac{\partial \mathrm{Var}(\hat{\beta}_{\lambda})}{\partial \lambda} &= \frac{\sigma^2_\nu \sigma^2_{\epsilon} F }{\alpha^4 n} \frac{1}{f^3} * \\
           &\ \bigg\{ \bigg[ 4 \frac{\lambda^3}{\gamma} v v' - 2 \frac{\lambda^3}{\gamma} v^3 - \frac{\lambda^3}{\gamma} v'' + \lambda^3 \left( 1+ \frac{\lambda}{\gamma}\right) v v'' - 2 \lambda^3 \left(1+\frac{\lambda}{\gamma} \right) v'^2 \bigg] * \lambda \\
          & - \bigg[ 4 \frac{\lambda^3}{\gamma} vv' - 2 \frac{\lambda^3}{\gamma} v^3 - \frac{\lambda^3}{\gamma} v'' + \lambda^3 \left(1 + \frac{\lambda}{\gamma}\right) vv'' - 2 \lambda^3 \left(1 + \frac{\lambda}{\gamma}\right) v'^2 \bigg] * \frac{1+\rho_{\epsilon \nu}^2}{F} \bigg\}
    \end{split}
\end{equation*}

We can now factor out $ \left( \lambda - \frac{1+\rho_{\epsilon \nu}^2}{F} \right)$ to get:
\begin{equation*}
    \begin{split}
          \frac{\partial \mathrm{Var}(\hat{\beta}_{\lambda})}{\partial \lambda} &= \frac{\sigma^2_\nu \sigma^2_{\epsilon} F }{\alpha^4 n} \bigg( \lambda - \frac{1+\rho_{\epsilon \nu}^2}{F} \bigg) * \\ 
          & \bigg\{ \frac{4 \frac{\lambda^3}{\gamma} vv' - 2 \frac{\lambda^3}{\gamma} v^3 - \frac{\lambda^3}{\gamma} v'' + \lambda^3 \left(1 + \frac{\lambda}{\gamma}\right) vv'' - 2 \lambda^3 \left(1 + \frac{\lambda}{\gamma}\right) v'^2 }{f^3} \bigg\}
    \end{split}
\end{equation*}

This shows that $\frac{\partial \mathrm{Var}(\hat{\beta}_{\lambda})}{\partial \lambda} = 0$ will be at least true for:

\begin{equation*}
    \lambda = \frac{1+\rho_{\epsilon \nu}^2}{F}
\end{equation*}

We can rewrite the asymptotic variance by plugging in this optimal tuning parameter:

\begin{equation}\label{variance_star}
    \begin{split}
         \mathrm{Var}(\hat{\beta}_{\lambda}) &= \frac{\sigma^2_\nu \sigma^2_{\epsilon}}{\alpha^4 n} \left[ \gamma F -\frac{\lambda \gamma (p' + p^2)}{f^2} \left(1 + \rho_{\epsilon \nu}^2 \right)  + \frac{\lambda^2(v' - v^2)}{f^2} \left(1 + \rho_{\epsilon \nu}^2 \right) \right] \\
         &= \frac{\sigma^2_\nu \sigma^2_{\epsilon}}{\alpha^4 n} \left[ \gamma F +\bigg( \frac{-\gamma \lambda (p' + p^2) + \lambda^2(v' - v^2)}{f^2} \bigg) \left(1 + \rho_{\epsilon \nu}^2 \right) \right]
    \end{split}
\end{equation} 

We use the following expressions for $v$ and $v'$:
\begin{equation}\label{formulas_v}
    \begin{split}
          v &= \frac{1}{\lambda} - \frac{\gamma}{\lambda} p \\
          v'&= \frac{1}{\lambda^2} - \frac{\gamma}{\lambda^2} p + \frac{\gamma}{\lambda} p'
    \end{split}
\end{equation}

and plug them into the asymptotic variance formula to get:

\begin{equation}
    \begin{split}
        \frac{\alpha^4 n}{\sigma^2_\nu \sigma^2_{\epsilon}} \mathrm{Var}(\hat{\beta}_{\lambda}) &=\gamma F + \frac{-\gamma \lambda (p'+p^2) + \gamma \lambda p' + 1 - \gamma p - 1 + 2 \gamma p - \gamma^2 p^2}{f^2} \left(1 +\rho_{\epsilon \nu}^2 \right) \\
        &=\gamma F + \frac{-\gamma \lambda p' - \gamma \lambda p^2 + \gamma \lambda p' + \gamma p - \gamma^2p^2}{f^2} \left(1 +\rho_{\epsilon \nu}^2 \right)\\
        &=\gamma F + \frac{\gamma p - \gamma^2 p^2 - \gamma \lambda p^2}{f^2} \left(1 +\rho_{\epsilon \nu}^2 \right)\\
        &=\gamma F + \frac{\gamma p[ 1-(\lambda+\gamma)p]}{f^2}  \left(1 +\rho_{\epsilon \nu}^2 \right)\\
        &= \gamma F + \frac{\gamma p f}{f^2}  \left(1 +\rho_{\epsilon \nu}^2 \right) \\
        &= \gamma F + \frac{\gamma p}{f}  \left(1 +\rho_{\epsilon \nu}^2 \right) \\
        &= \gamma F + \gamma \frac{1 - \lambda m}{f}  \left(1 +\rho_{\epsilon \nu}^2 \right)
    \end{split}
\end{equation}

This yields the asymptotic variance formula of the bias-adjusted 2SLS-Ridge estimator using the optimal tuning-parameter $\lambda^*$:

\begin{equation} \label{eq:var_lstar}
    \begin{split}
    \mathrm{Var}(\hat{\beta}_{\lambda^*}) = \frac{\sigma^2_\nu \sigma^2_{\epsilon}}{\alpha^4 n} \left[\gamma F + \gamma \frac{1 - \lambda^* m(-\lambda^*)}{f(-\lambda^*)}  \left(1 +\rho_{\epsilon \nu}^2 \right) \right]
    \end{split}
\end{equation}

\newpage
\section{Additional Simulation Results \label{appendix_simulation}}

\begin{figure}[h!]
	\centering
	\includegraphics[width=0.50\textwidth]{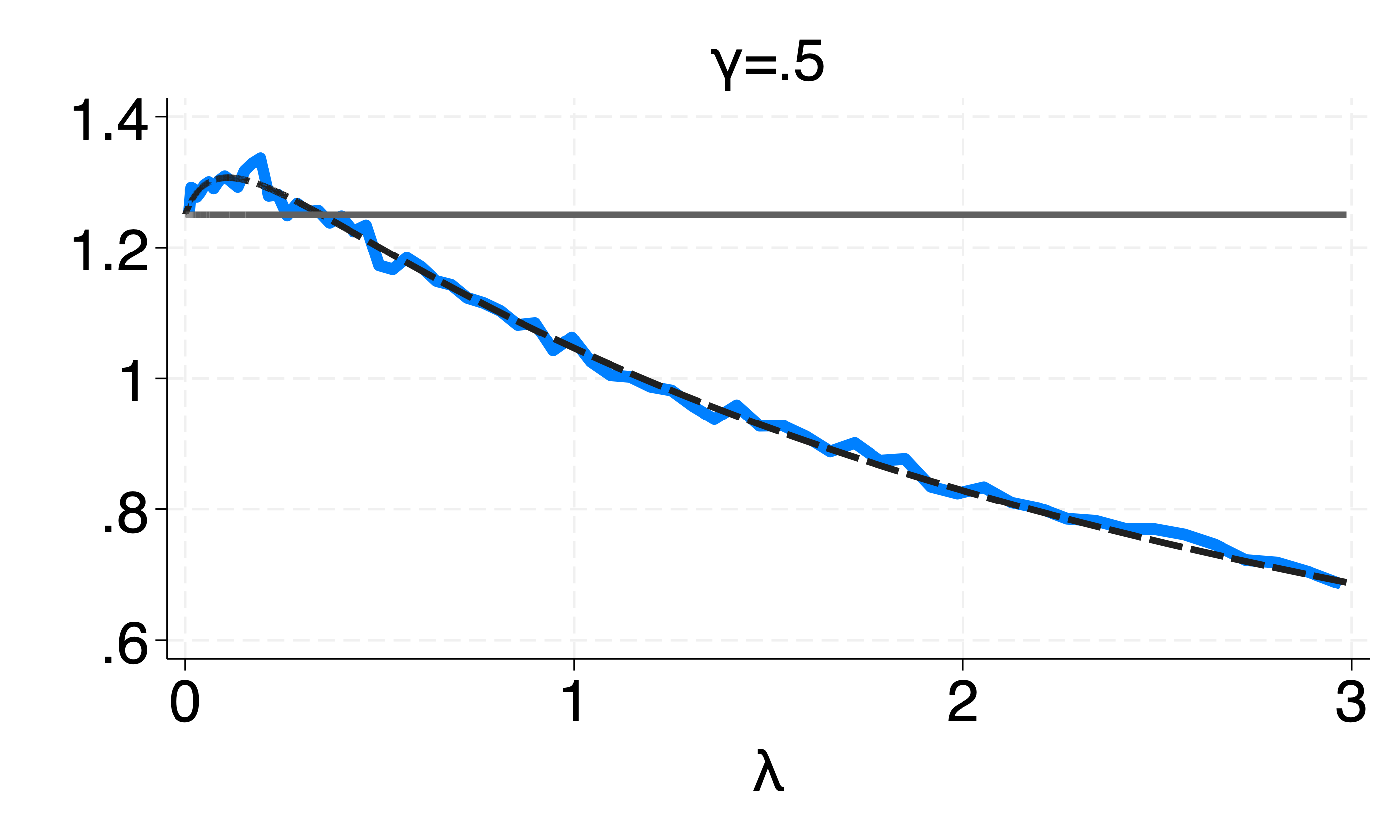}\hfill
	\includegraphics[width=0.50\textwidth]{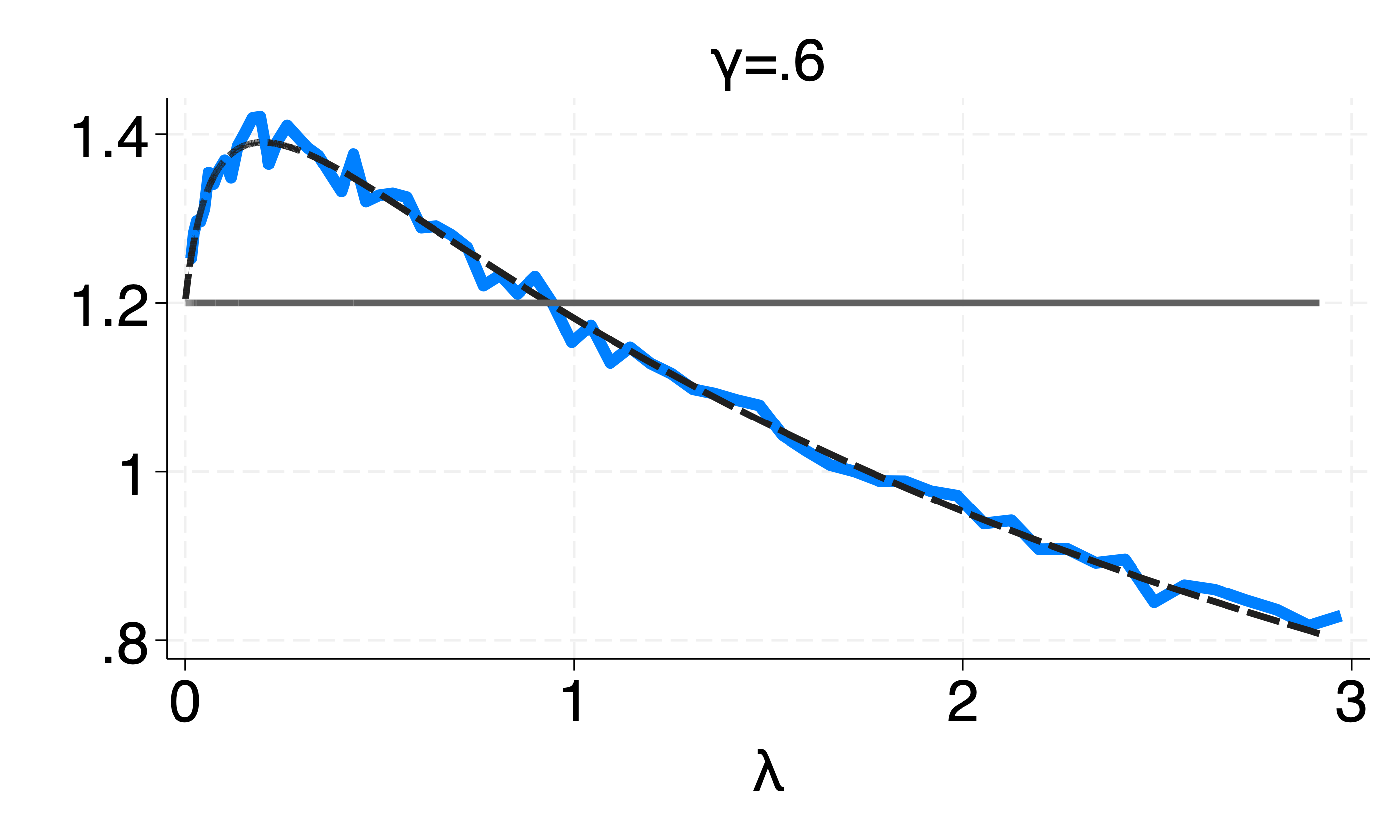}\hfill
	\includegraphics[width=0.50\textwidth]{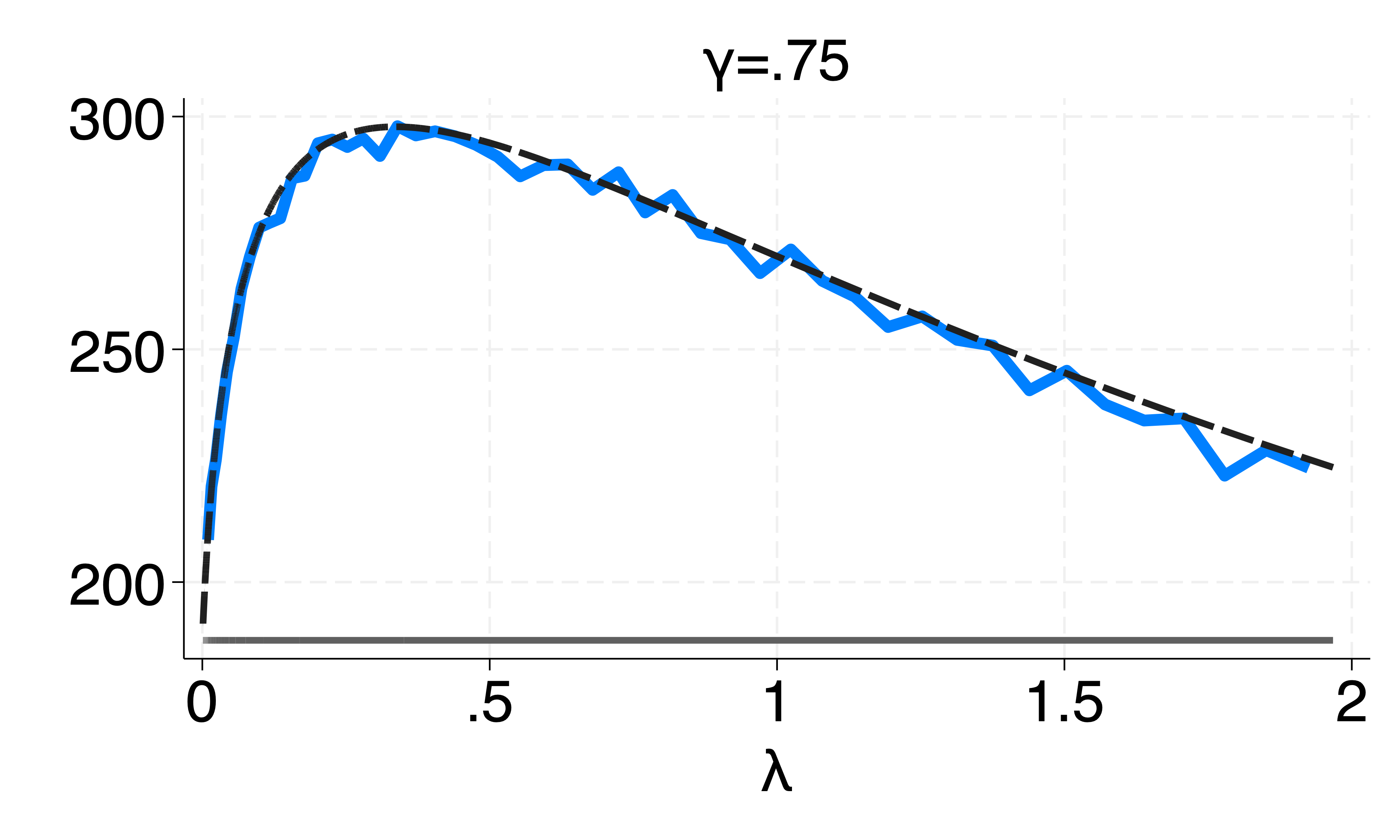}\hfill
	\includegraphics[width=0.50\textwidth]{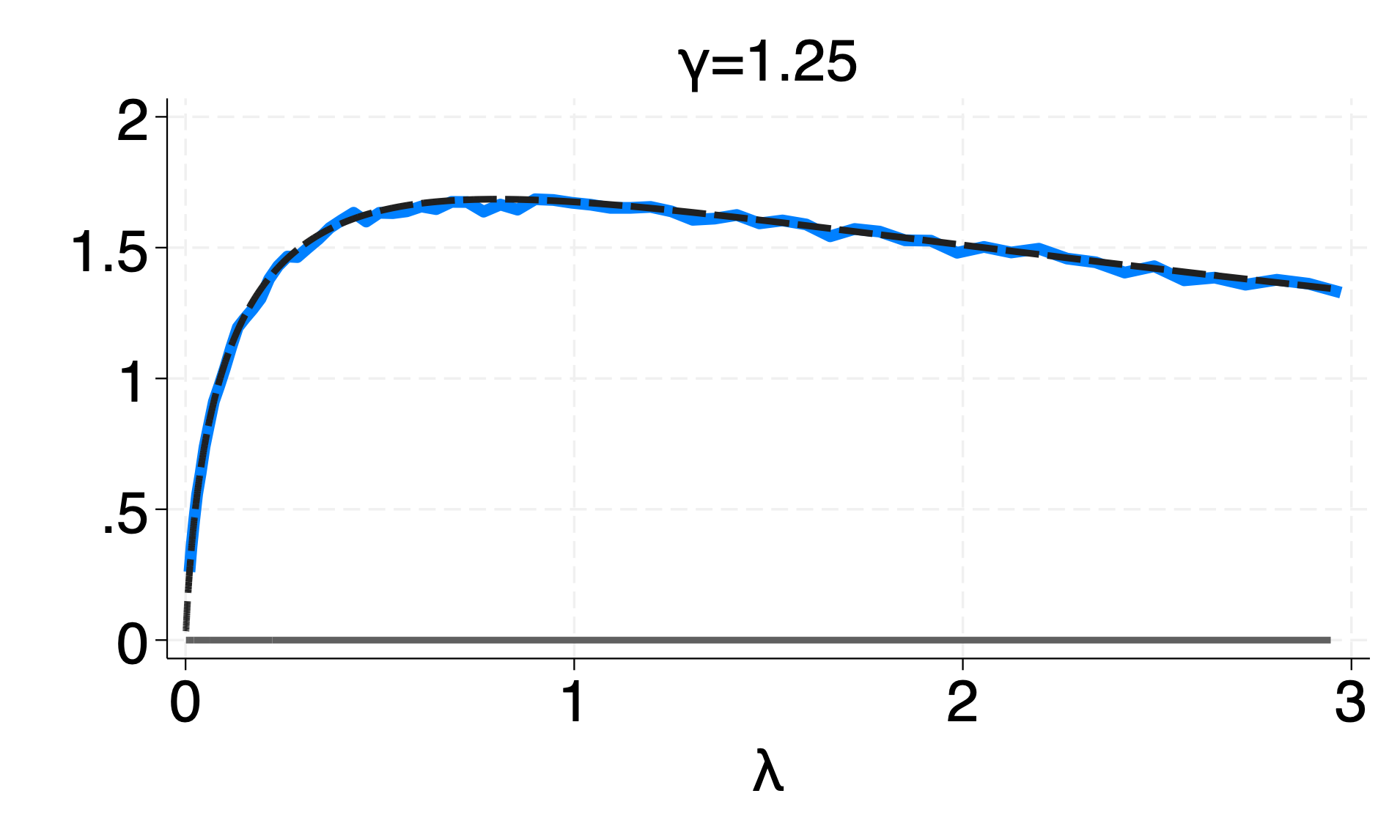}
	\caption{Signal approximation and simulation results for $x'S_\lambda x/n$ (average over 500 replications) and $\alpha^2 f\lam$ (dashed line); solid grey line indicates signal from standard bias-adjusted 2SLS, i.e. $\alpha^2(1-\gamma)$ if $\gamma<1$ and zero otherwise; AR-1 model $\Sigma_{ij}=0.5^{\vert i-j\vert}$. \label{figappendix_bias_signal}}
\end{figure}


\begin{figure}[h!]
	\centering
	\includegraphics[width=\textwidth]{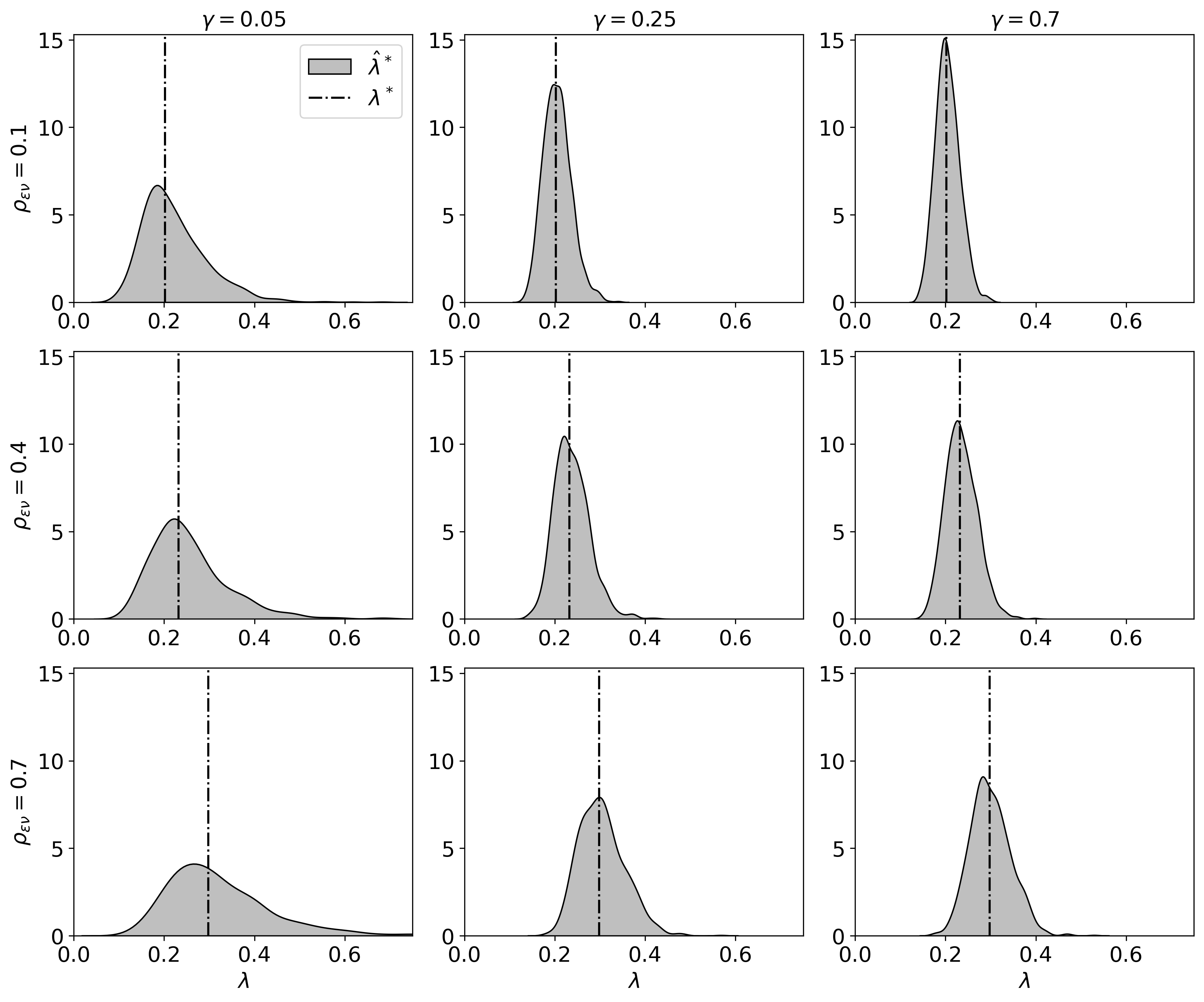}
	\caption{Density plots of $\hat{\lambda}^*$ obtained by estimating $\hat{F}$ and $\hat{\rho}_{\epsilon \nu}$ in a Monte-Carlo simulation with $n=1000$ and $1000$ repetitions, $F=5$, AR-1 model with $\rho_z = 0.5$; Column 1: $\gamma=0.05$; Column 2: $\gamma=0.25$; Column 3: $\gamma=0.7$; Row 1: $\rho_{\epsilon \nu}=0.1$ and $\lambda^*= 0.202$; Row 2: $\rho_{\epsilon \nu}=0.4$ and $\lambda^*= 0.232$; Row 3: $\rho_{\epsilon \nu}=0.7$ and $\lambda^*= 0.298$; Dash-dotted lines show for each setting the theoretical value of $\lambda^*$. \label{fig_lambda_star}}
\end{figure}

\begin{figure}[h!]
	\centering
	\includegraphics[width=0.50\textwidth]{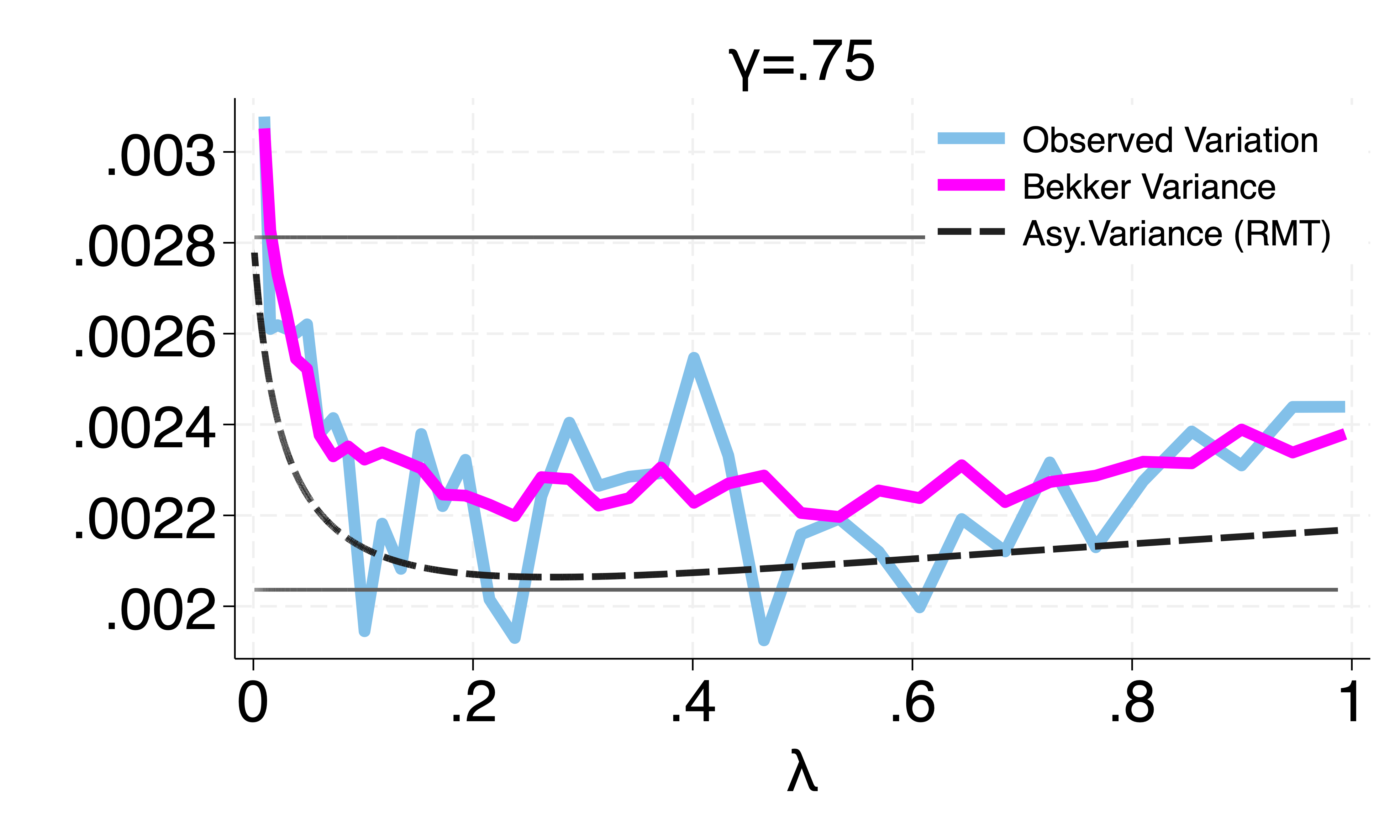}\hfill
	\includegraphics[width=0.50\textwidth]{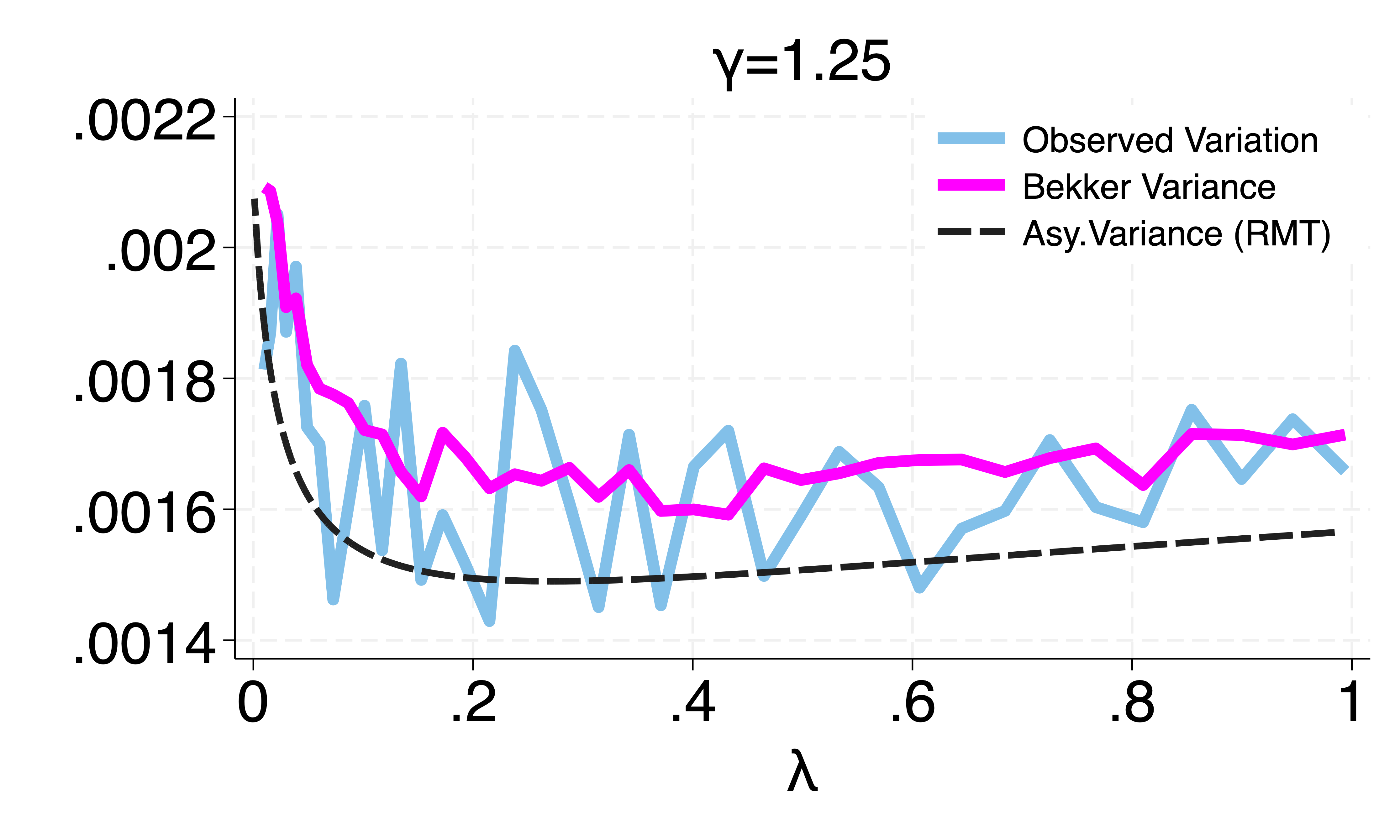}\hfill
	\includegraphics[width=0.50\textwidth]{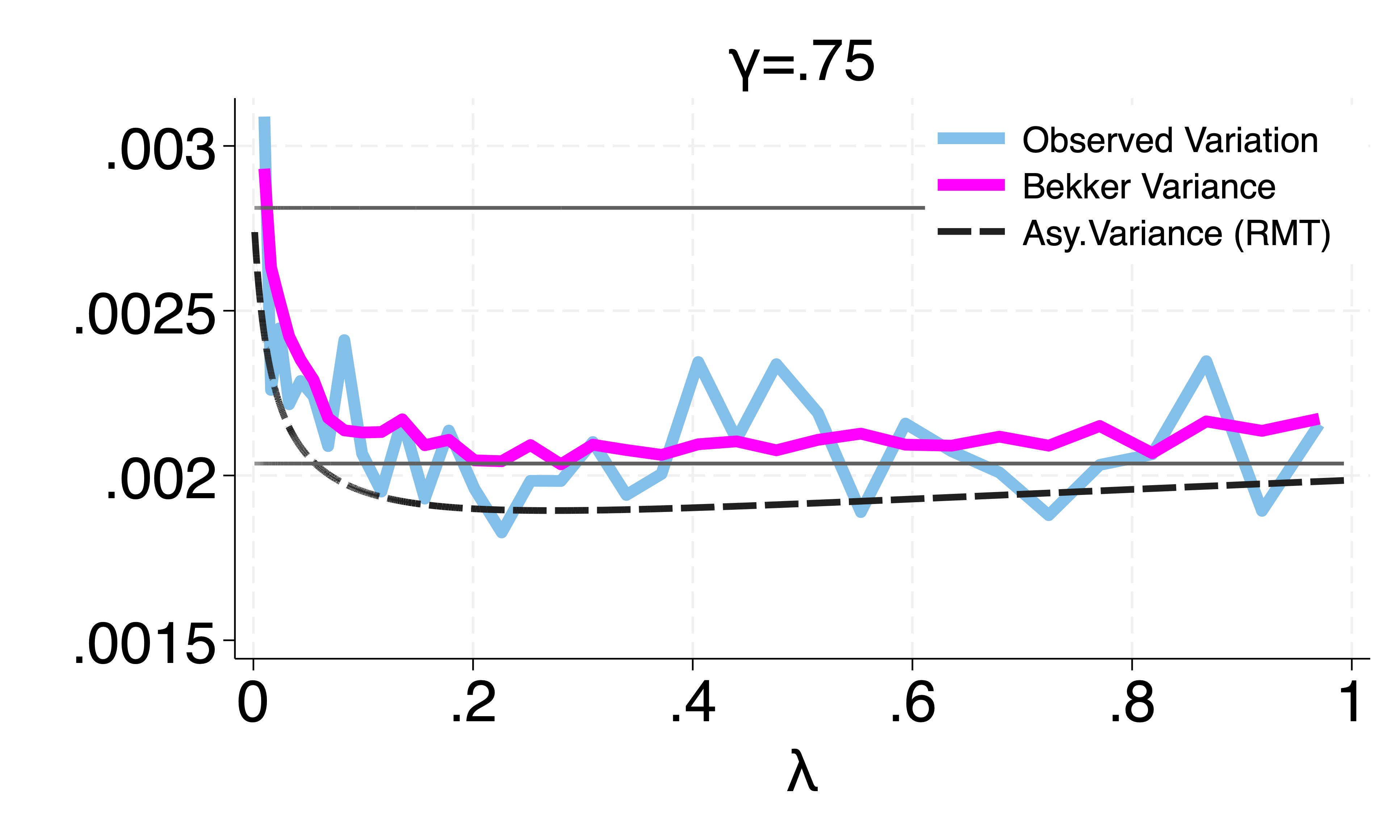}\hfill
	\includegraphics[width=0.50\textwidth]{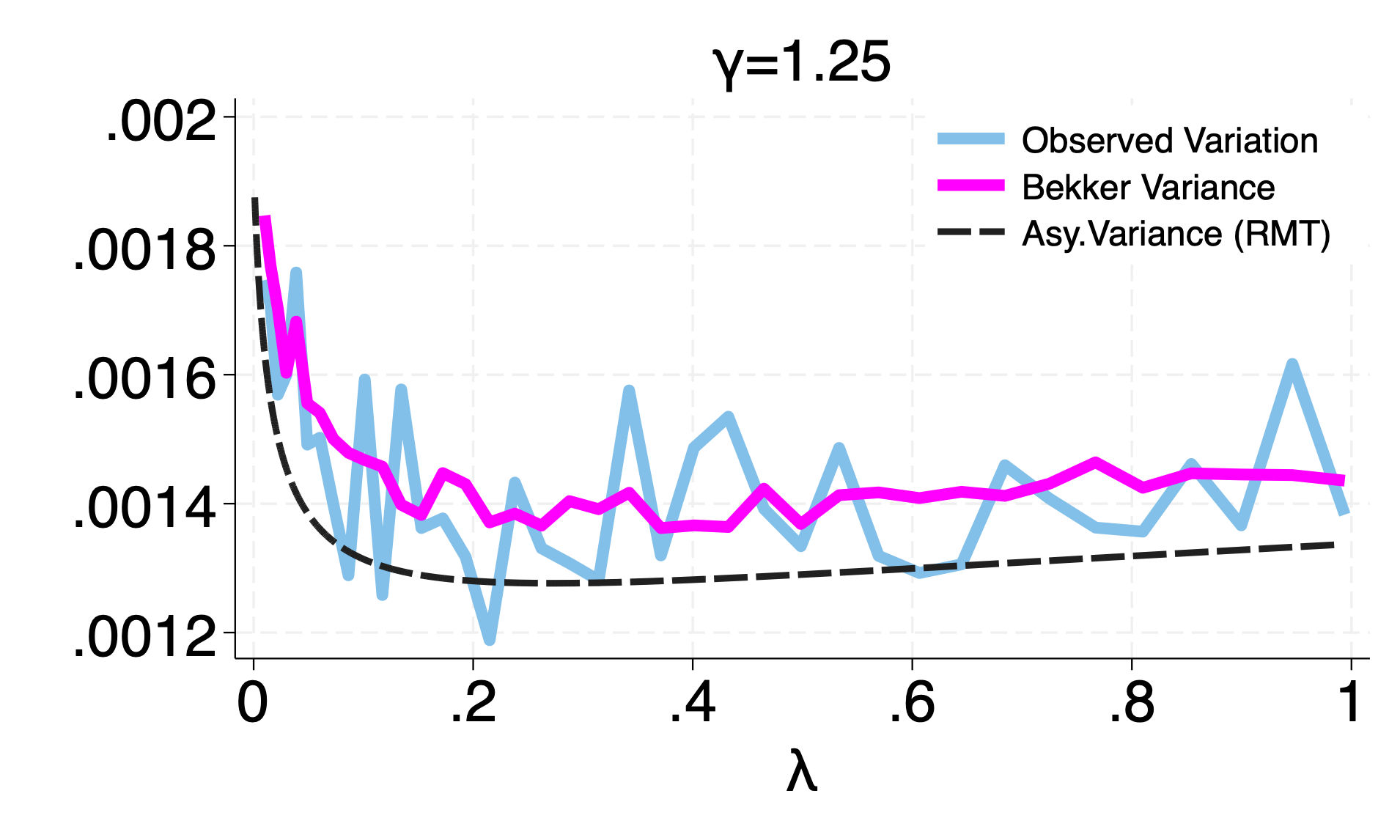}
	\caption{Observed variation and estimated variance of $\widehat\beta_\lambda$ (average over 500 replications) and asymptotic variance (dashed line) for $n=200$. The upper and lower horizontal lines depict the asymptotic variance of the standard bias-adjusted 2SLS and Liml, respectively; top: $\Sigma=I_k$; bottom: AR-1 model $\Sigma_{ij}=0.5^{\vert i-j\vert}$. \label{figappendix_variance}}
\end{figure}

\end{appendices}

\end{document}